\documentclass[11pt]{article} 
\usepackage[utf8]{inputenc} 
\usepackage[margin=1.0in]{geometry} 
\usepackage{graphicx} 
\usepackage{booktabs} 
\usepackage{array} 
\usepackage{paralist} 
\usepackage{verbatim} 
\usepackage{subfig} 
\usepackage{amsmath} 
\usepackage{amssymb} 
\usepackage{bbm} 
\usepackage{fancyhdr} 
\usepackage{amsthm}
\usepackage{authblk}
\numberwithin{equation}{section}
\usepackage{sectsty}
\allsectionsfont{\sffamily\mdseries\upshape}
\usepackage[nottoc,notlof,notlot]{tocbibind} 
\usepackage[titles,subfigure]{tocloft} 

\newtheorem{thm}{Theorem}
\newtheorem{lemma}{Lemma}
\newtheorem{cor}{Corollary}
\newcommand{\bea}{\begin{eqnarray}}
\newcommand{\eea}{\end{eqnarray}}
\newcommand{\lv}{\boldsymbol\lambda}

\newcommand{\Lv}{\boldsymbol L}
\newcommand{\Rv}{\boldsymbol R}
\newcommand{\C}{\mathbb{C}}
\newcommand{\E}{\mathbb{E}}

\newcommand{\tr}{\mbox{Tr}}
\newcommand{\G}{\mbox{Gin}(N,\mathbb{C})}
\newcommand{\xb}{\bar{x}}
\newcommand{\ind}{\mathbbm{1}}
\newcommand{\lambdab}{\bar{\lambda}}
\title{On the determinantal structure of conditional overlaps for the complex
Ginibre ensemble}
\author[1]{Gernot Akemann} 
\author[2]{Roger Tribe}
\author[2]{Athanasios Tsareas}
\author[2]{Oleg Zaboronski}
\affil[1]{
Faculty of Physics,
Bielefeld University,
D-33501 Bielefeld, Germany}
\affil[2]{Department of Mathematics, University of Warwick, CV4~7AL
Coventry, UK}
\date{} 
\begin{document}
\maketitle
\begin{abstract}
We continue the study of joint statistics of eigenvectors and eigenvalues initiated in the seminal papers
of Chalker and Mehlig. The principal object of our investigation is the expectation of the matrix of
overlaps between the left and the right eigenvectors for the complex $N\times N$ Ginibre ensemble, conditional 
on an arbitrary number $k=1,2,\ldots$ of 
complex
eigenvalues. These objects provide the simplest generalisation of the expectations
of the diagonal overlap ($k=1$) and the off-diagonal overlap ($k=2$) considered originally by Chalker and Mehlig.
They also appear naturally in the problem of joint evolution of eigenvectors and eigenvalues for Brownian
motions with values in complex matrices studied by the Krakow school.  

We find that these expectations possess a determinantal structure, where the relevant  kernels can be expressed
in terms of certain orthogonal polynomials in the complex plane. Moreover, the kernels admit a rather
tractable expression for all $N \geq 2$. This result enables a fairly straightforward calculation 
of the conditional expectation of the overlap matrix in the 
local
bulk and edge scaling limits as well as the proof
of the exact algebraic decay and asymptotic factorisation of these expectations 
in the bulk.
\end{abstract}
\section{Introduction
and Motivation.}
Let $\G$ be an ensemble of $N\times N$ matrices with independent complex Gaussian entries (complex
Ginibre ensemble): if $M\sim \G$ is a complex Ginibre matrix, then
\bea
\E(M_{ij})=0=\E(\bar{M}_{ij}),~1 \leq i,j\leq N,\\
\E(M_{ij}M_{kl})=0=\E(\bar{M}_{ij}\bar{M_{kl}}),~1 \leq i,j,k,l\leq N,\\
\E(M_{ij}\bar{M}_{kl})=\delta_{ik}\delta_{jl}~1 \leq i,j,k,l\leq N,
\eea
where `$~\bar{}~$' stands for complex conjugation. 
Let $\boldsymbol{\Lambda}^{(N)}=(\Lambda_1, \Lambda_2,\ldots, \Lambda_N)$
be the set of complex eigenvalues of $M$. This ensemble was introduced in 1965 in \cite{ginibre}
along with its real and quaternionic counterparts. It was immediately realised in this 
pioneering paper that the marginal distribution of eigenvalues for $\G$ can be computed
and is a natural generalisation of the corresponding answer for the Gaussian Unitary
Ensemble 
(GUE), cf. \cite{mehta}, 
to the case of 
a
complex spectrum:
\bea\label{cginibrelaw}
p_N\left(\boldsymbol{\Lambda}^{(N)} \in d\lv^{(N)}\right)d\lv^{(N)}=\frac{1}{Z_N} |\Delta^{(N)}(\lv^{(N)})|^2e^{-\sum_{j=1}^N |\lambda_j|^2}d\lv^{(N)}.
\eea
Here $p_N$ is the density for the distribution of the eigenvalues with respect to Lebesgue measure $d\lv^{(N)}=\prod_{k=1}^N d\lambda_k\lambdab_k$ on $\C^N$,
$\Delta^{(N)} (\lv^{(N)})=\prod_{i>j}^N(\lambda_i-\lambda_j)$ is the Vandermonde determinant, 
and
$Z_N= \pi^N\prod_{j=1}^N j!$
is the normalisation constant. As one can see, $p_N$ can be interpreted as the coordinate
distribution function of the classical $\log$-gas in two dimensions, which is parallel to the view of the ensemble of GUE eigenvalues
as the one-dimensional $\log$-gas. It is perhaps this analogy that biased the research into $\G$ towards the study
of its spectral properties, see e.g. 
\cite{mehta,forrester,KS}
for 
reviews of many significant results concerning the spectrum of $\G$. 
Yet, there is an important and nowadays widely appreciated difference between the complex
Ginibre ensemble and GUE:
even though the marginal distribution of eigenvalues for $\G$ can be 
computed analytically, the statistics
of eigenvectors does not decouple from the statistics of eigenvalues. 
Despite this fundamental difference,
it was not until the late nineties that Chalker and Mehlig initiated the
quantitative study of the joint
statistics of eigenvectors and eigenvalues for $\G$. They were motivated by questions of spectral stability
for non-Hermitian random Hamiltonians describing certain quantum or stochastic complex systems, 
a problem which can be traced back to \cite{may}. 
For further motivations from Physics see \cite{BurdaPRL}, \cite{Fyodorov_Savin} and references therein.

Stability questions can be hard to justify within a static setting of the complex Ginibre ensemble, but become
very natural if one considers some kind of dynamics (random or deterministic) on the space of complex matrices. 
Our own interest in the statistics of eigenvectors for $\G$ was inspired by the study of the stochastic dynamics of eigenvectors
and eigenvalues of complex matrices 
initiated by 
Z. Burda and M. Nowak,
their collaborators and students (`The Krakow school'), see 
e.g. \cite{Janik+,BurdaPRL}
and references therein.
So, in order to motivate the main subject of this paper, let us follow \cite{bourgade} and \cite{grela} and 
consider the Brownian motion
$M_t$ with values in $N\times N$ complex matrices started from zero. 
In other words, $(M_t)_{t\geq 0}$ is a Gaussian process with continuous paths, independent increments and the covariance
\bea
\E
\left(\tr\left(A^\dagger M_t^\dagger \right) 
\tr\left(M_s B\right)\right)
=\tr\left(A^\dagger B\right),
\eea
where $A,B$ are arbitrary $N\times N$ complex matrices. The fixed time $t>0$ marginal law for the process 
$(M_t)_{t \geq 0}$ coincides up to rescaling with the complex Ginibre ensemble.

Let $(\Lambda_{t\alpha}, \Lv_{t\alpha}, \Rv_{t\alpha})_{t\geq 0, 1\leq \alpha\leq N}$
be the induced processes, describing the evolution of eigenvalues of $M_t$ and the bi-orthogonal set of corresponding left and right eigenvectors, 
\bea
\Lv^\dagger_{t\alpha} M_t=\Lambda_{t\alpha} \Lv^\dagger_{t\alpha},~1\leq \alpha \leq N,\\
M_t \Rv_{t\alpha} =\Lambda_{t\alpha} \Rv_{t\alpha},~1\leq \alpha \leq N,\\
\langle \Lv_{t\alpha}, \Rv_{t\beta} \rangle=\delta_{\alpha,\beta},~1\leq \alpha,\beta \leq N,
\eea
where `$~\dagger~$' denotes Hermitean conjugation and $\langle\cdot,\cdot\rangle$ stands
for the Hermitean inner product on $\C^N$. As shown in 
\cite{bourgade} and \cite{grela}, the process $(\Lambda_{t \alpha})_{t \geq 0, 1\leq \alpha \leq N}$ is a complex martingale such that
\bea\label{martingale}
d\Lambda_{t\alpha} d\bar{\Lambda}_{t \beta}=O_{t\alpha \beta} dt,
\eea
where 
\bea\label{ovmat}
O_{t\alpha \beta}=\langle \Lv_{t\alpha}, \Lv_{t\beta} \rangle 
\langle \Rv_{t \alpha}, \Rv_{t \beta} \rangle,~1\leq \alpha,\beta \leq N
\eea
is the matrix of the overlaps between the left and the right eigenvectors of $M_t$. 
(It is worth noticing that paper \cite{grela} derives the full 
set of stochastic differential equations for the joint evolution of eigenvalues and eigenvectors of $M_t$  for 
any matrix size $N$.)
Notice that
for complex matrices, the matrix of overlaps is a non-trivial random variable as 
the left and the right eigenvectors are not orthogonal,
\[
\langle \Lv_\alpha, \Lv_\beta \rangle\neq 0, \langle\Rv_\alpha, \Rv_\beta\rangle \neq 0,~1\leq \alpha<\beta \leq N.
\]
As a result, the evolution of eigenvalues for complex matrices is very different from the case of normal matrices
with complex spectrum, despite both models having the same marginal 
distribution of eigenvalues for zero initial conditions.  See 
Appendix
\ref{appendix} for more details on the
dynamics of eigenvalues for normal matrices.

To study the evolution of eigenvalues corresponding to (\ref{martingale}), it is natural to study conditional expectations
\[
\E_N(d\Lambda_{t\alpha} d\bar{\Lambda}_{t \alpha}\mid \Lambda_{t\alpha}=\lambda_\alpha)=
\E_N(O_{t \alpha \alpha}\mid  \Lambda_{t\alpha}=\lambda_\alpha )dt,~ 
1\leq \alpha \leq N, 
\]
and
\[
\E_N(d\Lambda_{t\alpha} d\bar{\Lambda}_{t \beta}\mid \Lambda_{t \alpha}=\lambda_\alpha,
 \Lambda_{t\beta}=\lambda_\beta)=
 \E_N(O_{t\alpha \beta}\mid \Lambda_{t\alpha}=\lambda_\alpha,
 \Lambda_{t\beta}=\lambda_\beta)dt
 , ~1 \leq \alpha \neq \beta \leq N,
 \] 
 where $\E_N(\cdot)$ denotes expectation with respect to $\G$.
 These are the conditional
 expectations of the diagonal and the non-diagonal overlaps originally studied in \cite{chalker1,
 chalker2}.
Furthermore, if we wish to understand the influence of a fixed set of eigenvalues on the evolution 
of a single eigenvalue or a pair of eigenvalues, it is reasonable to consider general conditional expectations
\[
\E_N(O_{t\alpha_1 \alpha_2}\mid \Lambda_{t \alpha_p}=\lambda_{\alpha_p},
 p=1,2, \ldots k)
 , ~1 \leq \alpha_p \leq N, k=1,2,\ldots N.
\]

These are the principal objects studied in the present paper. An additional motivation for our study comes
from the mathematical structure of the answers: we find that conditional expectations of overlaps
are expressed in terms of determinants of matrices built out of a kernel of some integrable operator.
While this structure is 
a
well-known feature of point processes associated with the statistics
of eigenvalues of random matrices, we were unaware of determinantal answers for the
statistics of eigenvectors prior to starting our work. 

Our work continues 
the 
mathematical
study of the statistical properties of eigenvectors of non-Hermitian matrices, 
which has become an active research area during the past few years.
This renewed effort has already yielded a number of significant
generalisations of the original results by Chalker and Mehlig: In a breakthrough
paper \cite{bourgade}, Bourgade and Dubach prove that the law of the diagonal overlap conditional 
on the corresponding eigenvalue is given in the bulk
scaling limit by the inverse Gamma-distribution with parameter $2$. This is a significant generalisation
of the results of Chalker and Mehlig who managed to calculate this distribution for the matrix size $N=2$ only. The
statement follows from a beautiful novel representation of the diagonal overlap conditioned
on all eigenvalues  as a product of independent random variables. The authors also obtain new results
for the variance of the off-diagonal overlaps and the two-point function of diagonal overlaps, establishing in particular the
algebraic decay of the latter as a function of the distance between the corresponding eigenvalues. 

In
a parallel development
\cite{fyodorov}, Fyodorov obtains the full conditional law of the diagonal overlap both for
the complex Ginibre ensemble and the diagonal overlap associated with real eigenvalues for the 
real Ginibre ensemble. It is worth noting that the corresponding problem for the
real Ginibre overlaps associated with complex eigenvalues remains open. Fyodorov's answer 
is valid for $N<\infty$, which allows
him to derive the scaling limits for the distribution of the diagonal overlap both in the bulk and near
the edge of the spectrum as $N\rightarrow \infty$. Of course, the answers of \cite{fyodorov} are 
consistent with that of \cite{bourgade}. The calculations in \cite{fyodorov} are based on a
novel representation of the distribution of the diagonal overlap in terms  of ratios of determinants
and employs the calculus of anti-commuting variables, 
cf. \cite{FGS} for an alternative analytical approach. 

In \cite{walters}, Walters and Starr extend
the answers of \cite{chalker1,
chalker2} for the conditional expectation of the diagonal overlap at $N<\infty$ to any conditioned value
of the corresponding eigenvalue. This allows the authors to calculate
the edge scaling limit for the conditional expectation of the diagonal overlap. It is worth
stressing that our own calculations are based on the same analysis of recursion relations for the determinants of certain
$3$-diagonal moment matrices as in \cite{walters}. We complement it by an exact correspondence between diagonal and off-diagonal overlaps, which allows us to avoid difficulties associated with the analysis of the $5$-diagonal moment matrices.
In \cite{crawford}, Crawford and Rosenthal study high order moments of the overlap matrix (\ref{ovmat}). 
They prove the existence of the bulk scaling limit for the moments and discover a beautiful factorization
relation, 
valid on a macroscopic scale, expressing the moments of an arbitrary order in terms of a linear combination of products of moments
of order two, the structure of which deserves further investigation. 

In an investigation having a slightly
different flavour, the authors of \cite{luh} and \cite{rudelson} prove the delocalisation property of eigenvectors
for ensembles of complex random matrices, which do not necessarily possess unitary invariance.  The delocalisation
property means that the weight of the coefficients is not concentrated in any particular region of the 
index space. In \cite{zeitouni}, the authors study the statistics of angles between the eigenvectors
for 
invariant non-Gaussian
ensembles.

Finally, we must mention the work of the Krakow School, which is at least
partially responsible for the current renaissance of research into the joint statistics of eigenvectors and eigenvalues for
random non-Hermitian matrices. Among its recent contributions most relevant to the 
present work is the derivation of the system of 
stochastic  evolution equations for eigenvalues and eigenvectors,
cf. \cite{grela}, which allowed its authors to 
express the rate of change of eigenvalue correlation functions in terms of conditional expectations
of overlaps. These are precisely the objects studied in the present paper; in \cite{nowak}, its authors
present evidence for the 
microscopic
universality of moments of overlaps by exploiting a perturbative expansion in $N^{-1}$ 
for the calculation of moments for non-Gaussian ensembles of complex matrices.
In contrast, on a macroscopic scale the
eigenvector correlators explicitly depend on the radial spectral cumulative distribution and are thus  non-universal. This is shown in \cite{Belinschi}, combining free probability and the methods of generalised Green’s functions \cite{Janik+}.

The rest of the paper is organised as follows. Section \ref{results} presents our main results
concerning conditional expectations of overlaps: the determinantal representation for $N<\infty$,
the bulk and the edge scaling limits, exact algebraic asymptotic in the bulk for well separated eigenvalues.
Section \ref{proofs} contains the proofs
in the following subsections
: \ref{sec_setup}, \ref{sec_La1} the derivation of the determinantal representation for the
conditional expectations of diagonal and off-diagonal overlaps in terms
of bi-orthogonal polynomials in the complex plane; \ref{sec_heur} a heuristic calculation of the correlation
kernels, which shows how the result of rather complicated calculations of the following sections
can be easily guessed using the assumption of the extended translational invariance; \ref{sec_finiteN}
a rigorous evaluation of correlation kernels for $N<\infty$ in terms of the exponential polynomials;
\ref{proof_bulk} - \ref{proof_decay}
the calculation of  
various
scaling limits as $N\rightarrow \infty$. 
Appendix \ref{appendix}
contains the derivation of Dyson-like stochastic evolution equations for the normal matrix model. 

The methods used in the proofs are rather classical: the determinantal structure is a consequence of
Dyson's theorem reviewed in \cite{mehta} and the product structure of the overlap expectations 
conditioned on all eigenvalues; 
the computation of the correlation kernel for the diagonal overlaps reduces to the inversion
of the tri-diagonal moment matrix using the recursions already encountered in \cite{chalker1,
chalker2} and \cite{walters}; the
calculation of the kernel for the off-diagonals overlaps uses a relation between diagonal and off-diagonal overlaps
established in Lemma \ref{thm_rel} and 
determinantal 
identities, see \cite{tanner} for a review.
\section{Statement and Discussion of Results.
}\label{results}
As already explained in the 
introduction, we will be interested in the joint statistics of the overlaps and eigenvalues of $M \sim \G$.
Namely, we will study the following conditional expectations:
\bea
\E_N(O_{\alpha \alpha}\mid \Lambda_m=\lambda_m,m \in I),~I\subset \{1,2,\ldots,N\}, \alpha \in I\\
\E_N(O_{\alpha \beta}\mid \Lambda_m=\lambda_m, m \in J),~J\subset \{1,2,\ldots,N\}, \alpha, \beta \in J.
\eea
In other words , we consider the expectation of the overlaps with respect to $\G$ measure conditioned
on a set of eigenvalues. To be more concrete, if $M \sim \G$ is parametrised using Schur coordinates,
we compute the expected overlaps with respect to the product measure whose factors are the Haar measure
for the unitary conjugation, a
 Gaussian measure for the upper triangular degrees of freedom, and the eigenvalue measure obtained
 by conditioning (\ref{cginibrelaw}) on a set of eigenvalues. Due to the permutation symmetry of $\G$
 measure, it is sufficient to consider the following expectations:
\bea
\E_N(O_{11}\mid \Lambda_1=\lambda_1,\Lambda_2= \lambda_2, \ldots \Lambda_k=\lambda_k),~k=1,2,\ldots,N,\\
\E_N(O_{12}\mid \Lambda_1=\lambda_1,\Lambda_2= \lambda_2, \ldots \Lambda_k=\lambda_k),~k=2,\ldots,N.
\eea
Closely associated with these expectations are the following weighted multi-point intensities of the eigenvalues:
\bea\label{eq_d11}
D_{11}^{(N,k)}(\lv^{(k)}):=\E_N(O_{11}\mid \Lambda_1=\lambda_1
,\ldots, \Lambda_k=\lambda_k)\rho^{(N,k)}(\lv^{(k)}),\\~k=1,2,\ldots,N,\nonumber
\eea 
and 
\bea\label{eq_d12}
D_{12}^{(N,k)}(\lv^{(k)}):=\E_N(O_{12}\mid \Lambda_1=\lambda_1,
\ldots, \Lambda_k=\lambda_k)
\rho^{(N,k)}(\lv^{(k)}),\\~k=2,\ldots,N,\nonumber
\eea
where $\lv^{(k)}=(\lambda_1,\lambda_2, \ldots, \lambda_k)$ and $\rho^{(N,k)}$ is the $k$-point correlation function (Lebesgue density for factorial moments) for $\G$ eigenvalues. Recall that
\bea\label{cgin1}
\rho^{(N,k)}(\lv^{(k)}):=\frac{N!}{(N-k)!} \int_{\C^{N-k}} \prod_{m=k+1}^N d\lambda_m d\lambdab_m p_{N}(\lv^{(N)})=\det_{1\leq i,j\leq N}
\left(K^{(N)}_{ev}(\lambda_i,\lambda_j)
\right)
,
\eea
where 
\bea\label{cgin2}
K^{(N)}_{ev}(x,y)=\frac{1}{\pi}e^{-|x|^2}\sum_{m=0}^{N-1}\frac{(\bar{x}y)^m}{m!}
\eea
is the kernel of the determinantal point process corresponding to the distribution of $\G$ eigenvalues, see
\cite{ginibre} and \cite{mehta} for the derivation of (\ref{cgin1}) and (\ref{cgin2}).
\footnote{The kernel (\ref{cgin2}) can be re-written in a more symmetric form 
$\frac{1}{\pi}e^{-\frac{1}{2}(|x|^2+|y|^2)}\sum_{m=0}^{N-1}\frac{(\bar{x}y)^m}{m!}$ by the conjugation
$K_{ev}(x,y)\rightarrow e^{\frac{1}{2}x^2}K_{ev}(x,y)e^{-\frac{1}{2}y^2}$, which does not change the correlation functions.}

For the sake of brevity, we will refer to the expectations (\ref{eq_d11}) and (\ref{eq_d12}) as conditional overlaps. 
Notice that
\bea
D_{11}^{(N,1)}(\lambda)=\E_N\left(\sum_{\alpha=1}^N O_{\alpha \alpha}\delta(\Lambda_\alpha-\lambda)\right),\\
D_{12}^{(N,2)}(\lambda,\mu)=\E_N\left(\sum_{\alpha\neq \beta=1}^N O_{\alpha \beta}\delta(\Lambda_\alpha-\lambda)\delta(\Lambda_\beta-\mu)\right),
\eea
coincide with the expectations of diagonal and off-diagonal elements of the 
overlap matrix studied by Chalker and Mehlig: compare $D_{11}^{(N,1)}(\lambda)$ and 
$D_{12}^{(N,2)}(\lambda,\mu)$ with equations (10) and (11) 
of \cite{chalker2} evaluated at $\sigma=1$.

Our starting point is the fundamental result of \cite{chalker1,
chalker2} for the overlaps conditioned
on $all$ eigenvalues:
\bea\label{cmgenius11}
D_{11}^{(N,N)}(\lv^{(N)})=N!\prod_{k
=2}^N \left(1+\frac{1}{|\lambda_1-\lambda_k|^2}\right)p_N(\lv^{(N)}),\\
\label{cmgenius12}
D_{12}^{(N,N)}(\lv^{(N)})=- \frac{N!}{|\lambda_1-\lambda_2|^2}\prod_{k
=3}^N \left(1+\frac{1}{(\lambda_1-\lambda_k)
\left(\bar{\lambda}_2-\bar{\lambda}_k\right)}\right)p_N(\lv^{(N)}),
\eea
see equations (43) and (46) of \cite{chalker2}. Thus the task of integrating over the unitary and upper
triangular coordinates 
has already been accomplished by Chalker and Mehlig and we can concentrate on computing
the expectation of $D_{11}^{(N,N)}$, $D_{12}^{(N,N)}$ with respect to the conditional eigenvalue measure. 

The study of conditional expectations of overlaps is further simplified due to a simple relation between 
$D_{11}$
and $D_{12}$. In order to state this simple relation, we will treat $\{\lambda_i, \lambdab_i\}_{1\leq i\leq k}$
as independent complex variables and therefore treat conditional overlaps as functions on $\C^{2k}$. 
To obtain the final answer we will specialize to the real surface $\C^{k}\subset \C^{2k}$ by treating
$\lambda_i$ as the complex conjugate of $\lambdab_i$ for $1\leq i \leq k$. As a slight abuse of notation,
we will always refer to the value of the overlap at a point as $D_{12}^{(N,k)}(\lv^{(k)})$.
Let $\hat{T}$ be the following transposition acting on functions on $\C^{2k}$, $k \geq 2$:
\bea
\hat{T}f(\lambda_1, \bar{\lambda}_1, \lambda_2, \bar{\lambda}_2,\ldots)
=f(\lambda_1, \bar{\lambda}_2, \lambda_2, \bar{\lambda}_1,\ldots),
\eea
leaving the remaining variables $\lambda_3, \lambdab_3,\ldots \lambda_k,\lambdab_k$ untouched.
We have the following
\begin{lemma} (Exact relation between diagonal and off-diagonal overlaps for $N<\infty$.)\label{thm_rel}
For any $2\leq k \leq N<\infty$, the functions $D_{11}^{(N,k)}$ and $D_{12}^{(N,k)}$ are entire functions
on $\C^{2k}$. Moreover,
\bea\label{rltn}
D_{12}^{(N,k)}(\lv^{(k)})=- \frac{e^{-|\lambda_1-\lambda_2|^2}}
{1-|\lambda_1-\lambda_2|^2}\hat{T}D_{11}^{(N,k)}(\lv^{(k)}).
\eea
\end{lemma}
To state the main result of the paper we need to introduce some notations.
Let 
\begin{eqnarray}\label{expols}
e_{p}(x)=\sum_{k=0}^{p} \frac{x^k}{k!}, ~p=0,1,2,\ldots
\end{eqnarray}
be the exponential polynomial of order $p$ considered as functions on $\C$. Let
\bea\label{fpols}
f_{p}(x)=(p+1)e_{p}(x)-x e_{p-1}(x),~p=0,1,\ldots ,
\eea
where we define $e_{-1}(x)\equiv0$.
The polynomials $f_p$ are closely related to the bi-orthogonal polynomials 
in the complex plane associated with conditional overlaps, see Section \ref{sec_finiteN}
for details. 
Finally, let $\frak{F}_n: \C^3\rightarrow \C$ be the following polynomial in three variables:
\bea\label{total}
\frak{F}_n(x,y,z)&=&e_n(xy)\cdot e_n(xz)-e_n(xyz)\cdot e_n(x)\cdot \left(1-x(1-y)(1-z)\right)\\
&&+\frac{(1-y)(1-z)}{n!}\cdot \frac{(xyz)^{n+1}e_{n}(x)-x^{n+1}e_{n}(xyz)}{1-yz},~n=0,1,\ldots
\nonumber
\eea
The following is the main result of the paper:
\begin{thm} (Determinantal structure of conditional overlaps)\label{thm_fn}
For any $1\leq k\leq N<\infty$,
\bea\label{d11exact}
D_{11}^{(N,k)}(\lv^{(k)})=\frac{f_{N-1}(|\lambda_1|^2)}{\pi }
e^{-|\lambda_1|^2}
\det_{2\leq i,j\leq k}\left(K^{(N-1)}_{11}\left(\lambda_i,\bar{\lambda}_i,\lambda_j,\bar{\lambda}_j \mid
\lambda_1,\bar{\lambda}_1\right)\right)\label{ovlpn11},
\eea
where the kernel
\bea\label{thm1_kern}
K^{(N)}_{11}(x,\bar{x},y,\bar{y}\mid \lambda,\bar{\lambda})=\omega(x,\bar{x}\mid \lambda, \bar{\lambda})\kappa^{(N)}
(\xb,y\mid \lambda, \bar{\lambda}),
\eea
is a function on $\C^6$, which is built out of the weight
\bea
\omega(x,y\mid \lambda,\mu)=\frac{1}{\pi}(1+(x-\lambda)(y-\mu))e^{-xy},
\eea
a function on $\C^4$, and the reduced kernel
\bea\label{thm1_redkern}
\kappa^{(N)}
(\xb,y\mid \lambda, \bar{\lambda})=
\frac{\left(\left(N+1\right) \frak{F}_{N+1}\left(\lambda \bar{\lambda},\frac{\bar{x}}{\overline{\lambda}},\frac{y}{\lambda}\right)
-\lambda \bar{\lambda}\frak{F}_{N}\left(\lambda \bar{\lambda},\frac{\bar{x}}{\overline{\lambda}},\frac{y}{\lambda}\right) \right)}
{\left(\bar{x}-\overline{\lambda}\right)^2\left(y-\lambda\right)^2 f_{N}\left(\lambda \overline{\lambda}\right)}.
\eea
Furthermore, for $k\geq 2$,
\bea\label{thm_d12exact}
D_{12}^{(N,k)}(\lv^{(k)})&=&-\frac{e^{-|\lambda_1|^2-|\lambda_2|^2}}{\pi^2}f_{N-1}(\lambda_1\bar{\lambda}_2)
\kappa^{(N-1)}(\bar{\lambda}_1,\lambda_2\mid \lambda_1,\bar{\lambda}_2)
\nonumber\\
&&\times
\det_{3\leq i,j\leq k}\left(K^{(N-1)}_{12}\left(\lambda_i,\bar{\lambda}_i,
\lambda_j,\bar{\lambda}_j\mid
\lambda_1,\lambdab_1,\lambda_2,\bar{\lambda}_2\right)\right)\label{ovlpn12},
\eea
where
\bea\label{kern12}
K^{(N)}_{12}(x,\bar{x},y,\bar{y}\mid u,\bar{u},v,\bar{v})&=&
\frac{\omega(x,\bar{x}\mid u,\bar{v} )}{\kappa^{(N)}(\bar{u},v
\mid u,\bar{v} )}
\\
&&\times\det
\left(\begin{array}{cc}
\kappa^{(N)}(\bar{u},v\mid u,\bar{v} ) & \kappa^{(N)}(\bar{u},y\mid u,\bar{v} )\\
\kappa^{(N)}(\bar{x},v\mid u,\bar{v} ) & \kappa^{(N)}(\bar{x},y\mid u,\bar{v} )
\end{array}
\right).
\nonumber
\eea
\end{thm}
\noindent
{\bf Remark.} Everywhere in the paper we use the convention that the determinant of an empty matrix is equal to $1$.\\
\\
The finite-$N$ answer stated above enables an easy study of the large-$N$ limits of conditional overlaps. It is well known that the global spectral density of complex eigenvalues approaches the circular law, 
$\lim_{N\to\infty} \rho^{(N,1)}(\sqrt{N}z)= \frac{1}{\pi}\Theta(1-|z|)$, 
where $\Theta$ is the Heaviside step function,
cf. \cite{KS}. Therefore, 
we
will consider two such local, microscopic limits: the local bulk scaling limit,
\bea
D_{11}^{(bulk,~k)}(\lv^{(k)})&=&\lim_{N\rightarrow \infty} \frac{1}{N} D_{11}^{(N,k)}(\lv^{(k)}),\\
D_{12}^{(bulk,~k)}(\lv^{(k)})&=&\lim_{N\rightarrow \infty} \frac{1}{N} D_{12}^{(N,k)}(\lv^{(k)}),
\eea
i.e. we fix $\lambda_{1},\ldots,\lambda_{k}$ and take the large-$N$ limit, which places us in the vicinity of the origin\footnote{To access the general bulk we would have to scale $z=r e^{i\theta}\sqrt{N}+\lambda$, with $r<1$ and fixed $\lambda$.},
and the local edge scaling limit,
\bea\label{def_d11edge}
D_{11}^{(edge,~k)}(\lv^{(k)})&=&\lim_{N\rightarrow \infty}  \frac{1}{\sqrt{N}}D_{11}^{(N,k)}(e^{i \theta}
(\sqrt{N}+\lv^{(k)})),\\
\label{def_d12edge}
D_{12}^{(edge,~k)}(\lv^{(k)})&=&\lim_{N\rightarrow \infty} \frac{1}{\sqrt{N}}D_{12}^{(N,k)}(e^{i \theta}(\sqrt{N}
+\lv^{(k)})),
\eea
i.e. we shift to the vicinity of the spectral edge at $|z|=\sqrt{N}$, fix 
$\lambda_{1},\ldots,\lambda_{k}$ and then take the $N\rightarrow \infty$ limit.
The  overall rescaling of conditional overlaps used for the bulk and edge limits by the factors of $N^{-1}$ and $N^{-1/2}$
correspondingly is justified in the introduction.
Notice also that our notations for the edge scaling limit reflect the independence
of the final answer on the point at the edge of the spectrum around which we expand. 

\begin{cor}\label{thm_bulk}
(Local bulk scaling limit of conditional overlaps)
\bea
D_{11}^{(bulk,~k)}(\lv^{(k)})=
\frac{1}{\pi} \det_{2\leq i,j\leq k}
\left(K^{(bulk)}_{11}(\lambda_i,
\bar{\lambda}_i,
\lambda_j,\bar{\lambda}_j\mid \lambda_1,\bar{\lambda}_1)
\right)\label{ovlp11},
\eea
where $K^{(bulk)}_{11}:\C^6\rightarrow \C$ is the limiting kernel:
\bea\label{thm_k11bulk}
K^{(bulk)}_{11}(u,\bar{u},v,\bar{v}\mid \lambda,\bar{\lambda})=\omega^{(bulk)}(u,\bar{u}\mid \lambda, \bar{\lambda})
\kappa^{(bulk)}(\bar{u},v\mid \lambda, \bar{\lambda}),
\eea
where
\bea\label{thm_weight11bulk}
\omega^{(bulk)}(u,\bar{u}\mid \lambda, \bar{\lambda})=\frac{1}{\pi} (1+(u-\lambda)(\bar{u}-\bar{\lambda}))
e^{-(u-\lambda)(\bar{u}-\bar{\lambda})}
\eea
is the weight
and
\bea\label{thm_redker11bulk}
\kappa^{(bulk)}(\bar{u},v\mid \lambda, \bar{\lambda})=
\left.
\frac{d}{dz}\left(\frac{e^z-1}{z}\right)
\right|_{z=(\bar{u}-\bar{\lambda})(v-\lambda)},
\eea
is the reduced kernel. Moreover,
\bea\label{ovlp12}
D_{12}^{(bulk,~k)}(\lv^{(k)})=-\frac{1}{\pi^2} 
\kappa^{(bulk)}(\bar{\lambda}_1, \lambda_2\mid \lambda_1, \bar{\lambda}_2)
\det_{3\leq i,j\leq k}
\left(K^{(bulk)}_{12}(\lambda_i,
\bar{\lambda}_i,
\lambda_j,\bar{\lambda}_j\mid \lambda_1,\lambdab_1,\lambda_2,\bar{\lambda}_2)\right),\ \ 
\eea
where
\bea\label{eq_ker12bulk}
K^{(bulk)}_{12}(\lambda_i,
\bar{\lambda}_i,
\lambda_j,\bar{\lambda}_j\mid \lambda_1,\lambdab_1,\lambda_2,\bar{\lambda}_2)&=&
\frac{\omega^{(bulk)}(\lambda_i, \bar{\lambda}_i\mid \lambda_1, \bar{\lambda}_2)}
{\kappa^{(bulk)}(\bar{\lambda}_1, \lambda_2\mid \lambda_1, \bar{\lambda}_2)}\\
&\times&\det
\left(\begin{array}{cc}
\kappa^{(bulk)}(\bar{\lambda}_1, \lambda_2\mid \lambda_1, \bar{\lambda}_2)&
\kappa^{(bulk)}(\bar{\lambda}_1, \lambda_j\mid \lambda_1, \bar{\lambda}_2)\\
\kappa^{(bulk)}(\bar{\lambda}_i, \lambda_2\mid \lambda_1, \bar{\lambda}_2)&
\kappa^{(bulk)}(\bar{\lambda}_i, \lambda_j\mid \lambda_1, \bar{\lambda}_2)
\end{array}\right).\nonumber
\eea
\end{cor}
As expected, conditional overlaps in the bulk are translationally invariant, meaning that $D^{(bulk,~k)}_{11}$
and $D^{(bulk,~k)}_{12}$ are invariant with respect to 
a simultaneous shift of the arguments, 
\[
\lambda_i\rightarrow \lambda_i+\mu, \bar{\lambda}_i\rightarrow \bar{\lambda}_i+\bar{\mu},~1\leq i\leq k,~\mu \in \C.
\]
Less trivially, the overlaps in the bulk are invariant with respect to the above transformation for arbitrary complex numbers
$\mu$ and $\bar{\mu}$, which are not necessarily conjugate to each other. This extended translational invariance
is responsible for the success of the short heuristic derivation of Corollary \ref{thm_bulk} given in Section  \ref{sec_heur}.

Notice that the reduced kernel (\ref{thm_redker11bulk}) on the real line coincides with the density
of {\em eigenvalues} for a truncated unitary ensemble in the regime
of weak non-unitarity found by Sommers and Zyczkowski, see Eqn. (21) of \cite{sommers} at $L=1$. At the moment we do not understand any deep reason for such
a coincidence. 

Finally, let us verify that the statement of Corollary \ref{thm_bulk} agrees with Chalker and Mehlig's answer for
$D_{12}^{(bulk,2)}(\lambda_1,\lambda_2)$ obtained in \cite{chalker1,
chalker2}.
Specialising (\ref{ovlp12}) to the particular case
$k=2$ and denoting
\bea 
\lambda_{ij}=\lambda_i-\lambda_j,  ~\lambdab_{ij}=\lambdab_i-\lambdab_j,~1 \leq i,j\leq N,
\eea
we find that
\begin{eqnarray*}
D_{12}^{(bulk,~2)}(\lambda_1,\lambda_2)&=&
-\frac{1}{\pi^2} 
\kappa^{(bulk)}(\bar{\lambda}_1, \lambda_2\mid \lambda_1, \bar{\lambda}_2)
\\
&=&-\frac{1}{\pi^2}\frac{d}{dz}\left(\frac{e^z-1}{z}\right)\mid_{z=-|\lambda_{12}|^2}
=-\frac{1}{\pi^2}\frac{1}
{|\lambda_{12}|^4}\left(1-\left(1+|\lambda_{12}|^2\right)e^{-|\lambda_{12}|^2}\right),
\end{eqnarray*}
which corresponds to Eqn. (9) of \cite{chalker1} 
for fluctuations at the origin ($z_+=0$ in \cite{chalker1}).
We conjecture the corresponding local bulk kernels \eqref{thm_k11bulk} and \eqref{eq_ker12bulk} to be universal, see also \cite{nowak}.
\\

The finite-$N$ results stated in Theorem \ref{thm_fn} are also well suited for studying the statistics of overlaps
at the edge. For $a \in \C$, let
\bea\label{eq_erfc}
F(a)=\frac{1}{\sqrt{2\pi}}\int_a^\infty e^{-x^2/2}dx\equiv \frac{1}{2}\mbox{erfc}\left(\frac{a}{\sqrt{2}}\right),
\eea
where $\mbox{erfc}$ is the complementary error 
function, analytically continued to the complex plane. 
For any $a,b,c,d,f \in \C$, let
\bea\label{hedge}
&&
H(a,b,c,d,f)=-\frac{\sqrt{2\pi}}{\left(1-\sqrt{2\pi}ae^{\frac{a^2}{2}}F(a)\right)}\\\nonumber 
&&
\quad\quad\quad\quad\times
\left.
\frac{d}{dx}\left[e^{\frac{(a+x)^2}{2}}\left(e^{-f}F(b+x)F(c+x)-F(d+x)F(a+x)+fF(d)F(a+x)\right)\right]\right|_{x=0}.
\eea
\begin{cor}
(Local edge scaling limit of conditional overlaps)\label{thm_edge}
\bea\label{d11_edge}
D_{11}^{(edge,~k)}(\lv^{(k)})&=&
\frac{1}{\sqrt{2\pi^3}}\left(e^{-\frac{1}{2}(\lambda_1+\bar{\lambda}_1)^2}
-\sqrt{2\pi}(\lambda_1+\bar{\lambda}_1)F(\lambda_1+\bar{\lambda}_1)
\right) \nonumber \\ 
&&\times \det_{2\leq i,j\leq k}
\left(K^{(edge)}_{11}(\lambda_i,
\bar{\lambda}_i,
\lambda_j,\bar{\lambda}_j\mid \lambda_1,\bar{\lambda}_1)
\right),
\eea
where $K^{(edge)}_{11}:\C^6\rightarrow \C$ is the limiting kernel:
\bea\label{limkeredge}
K^{(edge)}_{11}(x,\bar{x},y,\bar{y}\mid \lambda,\bar{\lambda})=\omega^{(edge)}(x,\bar{x}\mid \lambda, \bar{\lambda})
\kappa^{(edge)}(\bar{x},y\mid \lambda, \bar{\lambda}),
\eea
where
\bea\label{limweightedge}
\omega^{(edge)}(x,\bar{x}\mid \lambda, \bar{\lambda})=\frac{1}{\pi} (1+(x-\lambda)(\bar{x}-\bar{\lambda}))
e^{-x\bar{x}}
\eea
is the weight
and
\bea\label{limredkeredge}
\kappa^{(edge)}(\bar{x},y\mid \lambda, \bar{\lambda})=e^{\bar{x}y}\frac{
H\left(\lambda+\bar{\lambda},\lambda+\bar{x},y+\bar{\lambda},y+\bar{x}, (\lambda-y)(\bar{\lambda}-\bar{x})\right)}
{(\lambda-y)^2(\bar{\lambda}-\bar{x})^2}
\eea
is the reduced kernel. Moreover, 
\bea\label{d12_edge}
D_{12}^{(edge,~k)}(\lv^{(k)})&=&
-\frac{1}{\sqrt{2\pi^5}}\left(1
-\sqrt{2\pi}(\lambda_1+\bar{\lambda}_2)e^{\frac{1}{2}(\lambda_1+\bar{\lambda}_2)^2}F(\lambda_1+\bar{\lambda}_2)
\right) \nonumber \\ 
&&\times \frac{e^{-|\lambda_1-\lambda_2|^2-\frac{1}{2}(\lambda_1+\bar{\lambda}_2)^2}}{\lambda_{12}^2\bar{\lambda}_{12}^2}
H(\lambda_1+\bar{\lambda}_2,\lambda_1+\bar{\lambda}_1,\lambda_2+\bar{\lambda}_2,
\lambda_2+\bar{\lambda}_1,-\lambda_{12}\bar{\lambda}_{12})
\nonumber\\
&&\times \det_{3\leq i,j\leq k}
\left(K^{(edge)}_{12}(\lambda_i,
\bar{\lambda}_i,
\lambda_j,\bar{\lambda}_j\mid \lambda_1,\lambdab_1,\lambda_2,\bar{\lambda}_2)\right),
\eea
where
\bea\label{eq_ker12edge}
K^{(edge
)}_{12}(\lambda_i,
\bar{\lambda}_i,
\lambda_j,\bar{\lambda}_j\mid \lambda_1,\lambdab_1,\lambda_2,\bar{\lambda}_2)&=&
\frac{\omega^{(edge)}(\lambda_i, \bar{\lambda}_i\mid \lambda_1, \bar{\lambda}_2)}
{\kappa^{(edge)}(\bar{\lambda}_1, \lambda_2\mid \lambda_1, \bar{\lambda}_2)}\\
&\times&\det
\left(\begin{array}{cc}
\kappa^{(edge)}(\bar{\lambda}_1, \lambda_2\mid \lambda_1, \bar{\lambda}_2)&
\kappa^{(edge)}(\bar{\lambda}_1, \lambda_j\mid \lambda_1, \bar{\lambda}_2)\\
\kappa^{(edge)}(\bar{\lambda}_i, \lambda_2\mid \lambda_1, \bar{\lambda}_2)&
\kappa^{(edge)}(\bar{\lambda}_i, \lambda_j\mid \lambda_1, \bar{\lambda}_2)
\end{array}\right).\nonumber
\eea
\end{cor}
As expected, the translational invariance is lost at the edge. However, it is easy to check that
$D^{(edge,~k)}_{11}$ and $D^{(edge,~k)}_{12}$ are invariant with respect to  a global shift along
the edge of the spectrum,
\[
\lambda_m\rightarrow \lambda_m+i\mu, \bar{\lambda}_m\rightarrow \bar{\lambda}_m-i\mu,~1\leq m\leq k,~\mu \in \mathbb{R}.
\]
This symmetry is just an infinitesimal version of the 
global $U(1)$-symmetry of the complex Ginibre ensemble, which survives in the large-$N$ limit. 

It follows from the statement of Corollary \ref{thm_edge}, that for $k=1$,
\bea\label{edge_d111}
D_{11}^{(edge,~1)}(\lambda_1)=\frac{1}{\sqrt{2\pi^3}}
\left( e^{-\frac{1}{2}(\lambda_1+\bar{\lambda}_1)^2}-\sqrt{2\pi}(\lambda_1+\bar{\lambda}_1)
F(\lambda_1+\bar{\lambda}_1)\right),
\eea
which coincides with the answer for the edge scaling limit of the diagonal overlap 
obtained in 
\cite[Corollary 4.3]{walters}. For $k=2$, we find that
\bea\label{edge_d122}
&&D_{12}^{(edge,~2)}(\lambda_1, \lambda_2)=\\
&&\frac{1}{\pi^2}
\frac{e^{-|\lambda_1-\lambda_2|^2-\frac{1}{2}(\lambda_1+\bar{\lambda}_2)^2}}{\lambda_{12}^2\bar{\lambda}_{12}^2}\nonumber
\frac{d}{dx}\bigg[e^{\frac{(\lambda_1+\bar{\lambda}_2+x)^2}{2}}
\bigg(e^{\lambda_{12}\bar{\lambda}_{12}}F(\lambda_1+\bar{\lambda}_1+x)F(\lambda_2+\bar{\lambda}_2+x)\\
\nonumber
&&
\left.
-F(\lambda_2+\bar{\lambda}_1+x)F(\lambda_1+\bar{\lambda}_2+x)
-\lambda_{12}\bar{\lambda}_{12}
F(\lambda_2+\bar{\lambda}_1)F(\lambda_1+\bar{\lambda}_2+x)\bigg)
\bigg]
\right|_{x=0},
\eea
which is apparently a new expression for the off-diagonal overlap at the edge. 
Again we conjecture the local edge kernels \eqref{limkeredge} and \eqref{eq_ker12edge} to be universal.

As it is easy to check, both the bulk and the edge scaling limits of $D_{11}^{(N,k)}$ 
and $D_{12}^{(N,k)}$ given 
in Corollaries \ref{thm_bulk} and \ref{thm_edge}
are related via the statement of Lemma \ref{thm_rel}. This reflects the fact that the large-$N$ limit 
preserves the analytic properties of conditional overlaps. In particular both the bulk and the edge scaling
limits of   $D_{11}^{(N,k)}$ and $D_{12}^{(N,k)}$ are entire functions of $\lv^{(k)}$ and $\bar{\lv}^{(k)}$.

There is also a different kind of relation between the scaling limits of overlaps: 
as we have already reviewed, the typical magnitude of the overlap in the bulk is $O(N)$, near the edge - $O(\sqrt{N})$. This is
consistent with the fact that the prefactor in (\ref{d11_edge}) diverges as we move back into the bulk:
if $Re(\lambda)=Re(\bar{\lambda})=R$,
\[
\lim_{R\rightarrow -\infty}\left(e^{-(\lambda+\bar{\lambda})^2}
-\sqrt{2\pi}(\lambda+\bar{\lambda})F(\lambda+\bar{\lambda})\right)=
-\sqrt{2\pi}\lim_{R\rightarrow -\infty}(\lambda+\bar{\lambda})=+\infty. 
\]
Therefore, there is no {\em a priori} reason for any relation between conditional overlaps in the bulk and at the edge.
However, simple analysis of the answers presented in Corollaries \ref{thm_bulk} and \ref{thm_edge}
reveals the following relations:
\begin{cor}\label{thm_edge_bulk}
\bea
&&\lim_{R\rightarrow -\infty} \frac{D^{(edge,~k)}_{11}(R\mathbf{1}^{(k)}+\lv^{(k)})}{D^{(edge,~1)}_{11}(R+\lambda_{1})}
=\frac{D^{(bulk,~k)}_{11}(\lv^{(k)})}{D^{(bulk,~1)}_{11}(\lambda_{1})},~k=1,2,\ldots,\\
&&\lim_{R\rightarrow -\infty} \frac{D^{(edge,~k)}_{12}(R\mathbf{1}^{(k)}+\lv^{(k)})}
{D^{(edge,~2)}_{12}(R+\lambda_{1},R+\lambda_{2})}
=\frac{D^{(bulk,~k)}_{12}(\lv^{(k)})}{D^{(bulk,~2)}_{12}(\lambda_{1},\lambda_{2})},~k=2,3,\ldots,
\eea
where $\mathbf{1}^{(k)}=(1,1,\ldots,1)\in \mathbb{R}^k$. 
\end{cor}
Notice that a similar relation for
the {\em eigenvalues} is known, see section 5.3. in \cite{honner}, but it is perhaps less surprising, as there is no
rescaling involved in the calculation of eigenvalue intensities in the bulk and at the edge.

Conditional overlaps provide a natural measure of dependence between eigenvectors and eigenvalues.
Recall, that for $\G$ the eigenvalue correlations decay exponentially with the square distance between the eigenvalues
on a large scale of separation, see
e.g. \cite{mehta}.
In contrast, the decay of correlations between eigenvalues and conditional overlaps is algebraic. 
\begin{cor}\label{thm_alg_dec}
(Exact algebraic asymptotic for conditional overlaps.)
Consider conditional overlaps $D_{11}^{(N,k)}$ and $D_{12}^{(N,k)}$ in the bulk scaling limit. Suppose
the eigenvalues $\lambda_1, \lambda_2, \ldots, \lambda_k$ are uniformly separated, 
i.e. 
there exists $L>0$:
\[
|\lambda_{ij}|\geq L, ~1\leq i<j\leq k. 
\]
Then, for large values of $L$,
\bea
D_{11}^{(bulk,k)}(\lv^{(k)})=\left(\frac{1}{\pi}\right)^{k}
\prod_{m=2}^k \left(1-\frac{1}{|\lambda_{m1}|^{4}}\right)+O(e^{-L^2}),\\
D_{12}^{(bulk,k)}(\lv^{(k)})
=-\left(\frac{1}{\pi}\right)^{k}\frac{1}{|\lambda_{12}|^4}
\prod_{m=3}^k \left(1-\frac{1}{\lambda_{m1}^2\bar{\lambda}_{m2}^2}
\right)+O(e^{-L^2}).
\eea
\end{cor}
Notice the that Corollary \ref{thm_alg_dec} implies an asymptotic factorisation
of conditional overlaps.  Namely, it establishes the existence of 
functions $P(\cdot \mid \lambda_1)$ and $Q(\cdot \mid \lambda_1, \lambda_2)$
on $\C$ such that 
\begin{eqnarray*}
\pi^k D_{11}^{(bulk, k)}(\lv^{(k)})=
\prod_{m=2}^kP(\lambda_m\mid \lambda_1)+O(e^{-L^2}),\\ 
\pi^k |\lambda_{12}|^4 D_{12}^{(bulk, k)}(\lv^{(k)})=-
\prod_{m=3}^kQ(\lambda_m\mid \lambda_1,\lambda_2)+O(e^{-L^2}).
\end{eqnarray*}
This statement is a 
consequence of a relation between conditional overlaps in the bulk  and correlation functions
for eigenvalues, which might be of independent interest:
\bea\label{prodd11}
D_{11}^{(bulk,~k)}(\lv^{(k)})&=&(-1)^{k-1} \prod_{m=2}^k \frac{1+|\lambda_{m1}|^2}{|\lambda_{m1}|^4}
\left(1-|\lambda_{m1}|^2-\lambda_{m1}\frac{\partial}{\partial \lambda_{m}}\right)\rho^{(bulk, k)}(\lv^{(k)}),\\
D_{12}^{(bulk,~k)}(\lv^{(k)})&=& \frac{(-1)^{k-1}}{|\lambda_{12}|^4}\left(1-\lambda_{21}\frac{\partial }{\partial \lambda_{2}}\right)\prod_{m=3}^k \frac{1+\lambda_{m1}\lambdab_{m2}}{\lambda_{m1}^2\lambdab_{m2}^2}
\left(1-\lambda_{m1}\lambdab_{m2}-\lambda_{m1}\frac{\partial}{\partial \lambda_{m}}\right)
\nonumber\\
&&\times\rho^{(bulk,k)}(\lv^{(k)}).
\label{prodd12}
\eea
Notice that the differential operators entering the product in the right hand side of (\ref{prodd11}) and (\ref{prodd12}) commute, so there is no ambiguity in the above formulae due to the ordering, see Section \ref{proof_decay} for the derivation.
We conjecture that the algebraic decay and the factorisation property for the conditional
overlaps stated in Corollary \ref{thm_alg_dec}
remain true in the global bulk scaling limit as well.


\section{Proofs}\label{proofs}
\subsection{General set-up for the proof of Theorem \ref{thm_fn}. The determinantal
structure.}\label{sec_setup}
Recall expressions  (\ref{cmgenius11}) and (\ref{cmgenius12}) for  
the overlaps conditioned on $N$ eigenvalues.
Averaging over all the eigenvalues but $\lambda_1, \ldots, \lambda_k$, we get
\bea
D_{11}^{(N,k)}(\lv^{(k)})&=&\frac{1}{Z_N}\frac{N!}{(N-k)!}\int_{\C^{N-k}}\prod_{i=k+1}^{N}d\lambda_i
d\bar{\lambda}_i
|\Delta^{(N)}(\lambda_1,\lambda_2,\ldots, \lambda_N)|^2 e^{-\sum_{j=1}^N |\lambda_j|^2}
\nonumber\\
&&
\quad\quad\quad\quad\times 
\prod_{\ell=2}^N \left(1+\frac{1}{|\lambda_1-\lambda_\ell|^2}\right),\\
D_{12}^{(N,k)}(\lv^{(k)})&=&\frac{1}{Z_N}\frac{N!}{(N-k)!}\int_{\C^{N-k}}\prod_{i=k+1}^{N}
d\lambda_i
d\bar{\lambda}_i
|\Delta^{(N)}(\lambda_1,\lambda_2,\ldots, \lambda_N)|^2 e^{-\sum_{j=1}^N |\lambda_j|^2}
\nonumber\\
&&\quad\quad\quad\quad\times
\frac{1}{|\lambda_1-\lambda_2|^2}\prod_{\ell=3}^N \left(1+\frac{1}{(\lambda_1-\lambda_\ell)
\left(\bar{\lambda}_2-\bar{\lambda}_\ell\right)}\right),
\eea
where $Z_N=\pi^N \prod_{j=1}^N j!$ is the normalisation constant. 
Therefore, 
\bea
D_{11}^{(N,k)}(\lv^{(k)})&=&\frac{e^{-|\lambda_1|^2}}{Z_N}\frac{N!}{(N-k)!}\int_{\C^{N-k}}\prod_{i=k+1}^{N}d\lambda_i
d\bar{\lambda}_i
|\Delta^{(N-1)}(\lambda_2,\ldots, \lambda_N)|^2 \nonumber\\
&&\times \prod_{m=2}^N \pi \omega( \lambda_m,\bar{\lambda}_m\mid \lambda_1,\bar{\lambda}_1),
\label{eqnd11int}\\
D_{12}^{(N,k)}(\lv^{(k)})&=&-\frac{e^{-|\lambda_1|^2-|\lambda_2|^2}}{Z_N}\frac{N!}{(N-k)!}\int_{\C^{N-k}}
\prod_{i=k+1}^{N}d\lambda_i
d\bar{\lambda}_i
\Delta^{(N-1)}(\lambda_2,\lambda_3,\ldots, \lambda_N)\nonumber\\
&&\times
\Delta^{(N-1)}(\bar{\lambda}_1,\bar{\lambda}_3,\ldots, \bar{\lambda}_N)
\prod_{m=3}^N\pi \omega(\lambda_m, \bar{\lambda}_m\mid \lambda_1, \bar{\lambda}_2),\label{eqnd12int}
\eea
where the integration measure is defined in both cases by the following function on $\C^3$:
\bea\label{eqnwt}
\omega(z,\bar{z}|u,v)=\frac{1}{\pi}(1+(z-u)(\bar{z}-v))e^{-z\bar{z}},
~z,u,v \in \C.
\eea

In order to determine $D_{12}^{(N,k)}$ using Lemma \ref{thm_rel}, proved in Section \ref{proof_of_lemma_1} below, we need to calculate  $D_{11}^{(N,k)}$  treating 
the 
complex variables $\lv^{(k)}$ and $\bar{\lv}^{(k)}$ as independent. 
The first steps are  standard, see 
e.g.  \cite{mehta}.
Using elementary linear algebra,
\bea\label{det}
&&|\Delta^{(N-1)}(\lambda_2,\ldots, \lambda_N)|^2
\prod_{m=2}^N \omega(\lambda_m,\bar{\lambda}_m\mid \lambda_1,\bar{\lambda}_1)
\nonumber\\
&=&\prod_{q=0}^{N-2}\langle P_q , Q_q \rangle
 \det_{2\leq i,j\leq N}
 \left(K^{(N-1)}_{11}(\lambda_i,\bar{\lambda}_i,\lambda_j,\bar{\lambda}_j\mid
\lambda_1, \bar{\lambda}_1 )\right),
\eea
where $K^{(N)}_{11}$ is the following kernel (of an integral operator):
\bea\label{eq_kerfn}
K^{(N)}_{11}(x,\bar{x},y,\bar{y}\mid \lambda_1, \bar{\lambda}_1)=\sum_{k=0}^{N-1} 
\frac{\overline{P_{k}(x)} Q_{k}(y)}{\langle P_{k}, Q_k\rangle}
\omega(x, \bar{x}\mid \lambda_1, \bar{\lambda}_1),
\eea
and $\{P_i,Q_i\}_{i=0}^{\infty}$ are holomorphic monic
polynomials on $\C$, bi-orthogonal with respect to the weight $\omega(\cdot,\cdot \mid 
\lambda_1, \bar{\lambda}_1)$:
\bea\label{eq_bop}
\langle P_i, Q_j\rangle :=\int_{\C} dzd\bar{z} \omega(z, \bar{z}
\mid \lambda_1, \bar{\lambda}_1)
\overline{P_i(z)} Q_j(z)=\langle P_i,Q_i\rangle \delta_{i,j},~0\leq i,j<\infty.
\eea
Notice that the bi-orthogonal polynomials depend on $\lambda_1$ and $\bar{\lambda}_1$ as parameters,
but we will suppress this dependence in order to simplify the notation. 
We will establish the existence of the bi-orthogonal polynomials and the associated kernel (\ref{eq_kerfn})
for the concrete weight $\omega$ by constructing
them explicitly, for a general discussion see \cite{akemann_vernizzi}.

In what follows it will  be convenient to define the reduced kernel $\kappa^{(N)}$ via
\bea
K^{(N)}_{11}(x,\bar{x},y,\bar{y}\mid \lambda_1, \bar{\lambda}_1)&=&\kappa^{(N)}(\bar{x},y\mid \lambda_1, \bar{\lambda}_1) \omega(x, \bar{x}\mid \lambda_1, \bar{\lambda}_1)\nonumber,\\
\kappa^{(N)}(\bar{x},y\mid \lambda_1, \bar{\lambda}_1)&=&
\sum_{k=0}^{N-1} \frac{\overline{P_{k}(x)} Q_{k}(y)}{\langle P_k, Q_{k} \rangle}.
\label{eq_redk}
\eea
Notice that the kernel $K^{(N)}_{11}$ is self-reproducing, 
\[
K^{(N)}_{11}*K^{(N)}_{11}=K^{(N)}_{11}.
\]
(Equivalently, the corresponding integral operator acting on polynomials
is a projection.) 
Therefore, Dyson's theorem is applicable to the calculation of the integral in (\ref{eqnd11int}).
\footnote{The self-adjointness of a kernel is not necessary for the applicability of Dyson's theorem.}
 Substituting (\ref{det})
into (\ref{eqnd11int}) and applying the theorem, we find that
\bea\label{eqd11_gen}
D_{11}^{(N,k)} (\lv^{(k)})&=&\frac{\pi^{N-1} N!}{Z_N} \prod_{q=0}^{N-2}
\langle P_q ,Q_q \rangle\cdot
e^{-|\lambda_1|^2}\nonumber\\
&&\times\det_{2\leq i,j\leq k}
\left(K^{(N-1)}_{11}(\lambda_i,\bar{\lambda}_i,\lambda_j,\bar{\lambda}_j\mid \lambda_1,\bar{\lambda}_1)\right).
\eea
Observe the emergence of the determinantal structure for the diagonal conditional overlaps.
The off-diagonal overlap $D_{12}^{(N,k)}$ as a function on $\C^{2k}$ can now be computed using Lemma \ref{thm_rel}:
\bea
D_{12}^{(N,k)}(\lv^{(k)})&=&-\frac{\pi^{N-1}N!}{Z_N}
\frac{e^{-\lambda_{12}\bar{\lambda}_{12}-\lambda_1 \bar{\lambda}_2}}{1-\lambda_{12}\bar{\lambda}_{12}}
\hat{T}\left(
\prod_{q=0}^{N-2}\langle P_q , Q_q \rangle\right)
\nonumber\\
&&\times \det \left(
\begin{array}{c|c}
 K^{(N-1)}_{11}(\lambda_2,\bar{\lambda}_1,\lambda_2,\bar{\lambda}_1\mid \lambda_1,\bar{\lambda}_2)&
 K^{(N-1)}_{11}(\lambda_2,\bar{\lambda}_1,\lambda_j,\bar{\lambda}_j\mid \lambda_1,\bar{\lambda}_2), \\
 ~&3\leq j\leq k\\
 \hline
 ~&~\\
 K^{(N-1)}_{11}(\lambda_i,\bar{\lambda}_i,\lambda_2,\bar{\lambda}_1\mid \lambda_1,\bar{\lambda}_2),&
 K^{(N-1)}_{11}(\lambda_i,\bar{\lambda}_i,\lambda_j,\bar{\lambda}_j\mid \lambda_1,\bar{\lambda}_2),\\
 3\leq i\leq k&3\leq i,j\leq k
\end{array}\right).\nonumber\\
\label{D12Nk}
\eea 
It is worth stressing that $\prod_{q=0}^{N-2}\langle P_q , Q_q \rangle$
 is a function of $\lambda_1, \lambdab_1$, therefore the action
of $\hat{T}$ on this product is non-trivial. Recall also that 
$\lambda_{ij}:=\lambda_i-\lambda_j$, $\lambdab_{ij}:=\lambdab_i-\lambdab_j$.
The determinant in the above formula can be re-written using the 
following determinantal identity
\bea\label{eq_ti}
\det_{1\leq i,j\leq n}(a_{ij})=a_{11}\det_{2\leq i,j\leq n}
\left(
a_{11}^{-1}
\det\left(\begin{array}{cc}
a_{11} & a_{1j}\\
a_{i1} & a_{ij}
\end{array}
\right)\right),~a_{11}\neq 0.
\eea
This follows from a well known identity for block determinants, see e.g. \cite{sylv}:
\begin{equation}
\det\left(
\begin{array}{cc}
A&B\\
C&D\\
\end{array}
\right)=\det(A)\det\left(D-CA^{-1}B\right),
\label{blockId}
\end{equation}
valid for invertible matrices $A$. Namely, choosing for $A=a_{11}\neq0$ in \eqref{eq_ti} we have 
\begin{equation}
\det_{2\leq i,j\leq N}\left(
\begin{array}{cc}
a_{11}&a_{1j}\\
a_{i1}&a_{ij}\\
\end{array}
\right)
=\det(a_{11})\det_{2\leq i,j\leq N}\left(a_{ij}-a_{i1}a_{11}^{-1}a_{1j}\right)
=a_{11}\det_{2\leq i,j\leq N}\left(a_{11}^{-1}(a_{11}a_{ij}-a_{i1}a_{1j})\right).
\end{equation}
Eq. \eqref{eq_ti} can be seen as the simplest of  Tanner's identities for determinants and Pfaffians, see 
e.g. \cite{tanner} for a review.
Applying \eqref{eq_ti} to \eqref{D12Nk} results into
\bea\label{eqd12_gen}
&&D_{12}^{(N,k)}(\lv^{(k)})=-\frac{\pi^{N-2}N!}{Z_N}
 \hat{T}\left(
\prod_{q=0}^{N-2}\langle P_q ,Q_q \rangle\right)
e^{-\lambda_{1}\bar{\lambda}_{1}-\lambda_2 \bar{\lambda}_2}
\kappa^{(N-1)}\left(\bar{\lambda}_1,\lambda_2|\lambda_1, \bar{\lambda}_2\right)
\nonumber\\
&&\times \det_{3\leq i,j\leq k} \left(\frac{\omega(\lambda_i,\bar{\lambda}_i\mid \lambda_1,\bar{\lambda}_2)}
{\kappa^{(N-1)}\left(\bar{\lambda}_1,\lambda_2|\lambda_1, \bar{\lambda}_2\right)}
\det\left(
\begin{array}{cc}
\kappa^{(N-1)}(\bar{\lambda}_1,\lambda_2\mid \lambda_1,\bar{\lambda}_2)&
\kappa^{(N-1)}(\bar{\lambda}_1,\lambda_j\mid \lambda_1,\bar{\lambda}_2) \\
\kappa^{(N-1)}(\bar{\lambda}_i,\lambda_2\mid \lambda_1,\bar{\lambda}_2)&
\kappa^{(N-1)}(\bar{\lambda}_i,\lambda_j\mid \lambda_1,\bar{\lambda}_2)
\end{array}\right)\right),
\nonumber\\
\eea 
which explains the structure of the claim 
 (\ref{ovlpn12}), (\ref{kern12}) of Theorem \ref{thm_fn}.
The final answers for the conditional overlaps are obtained by evaluating (\ref{eqd11_gen})
and (\ref{eqd12_gen}) on the real surface $\C^{k}\subset \C^{2k}$, specified 
by the equations $\lv^{(k)}=\overline{\bar{\lv}^{(k)}}$.

The proof of Theorem \ref{thm_fn} is therefore reduced to the calculation of the reduced kernel $\kappa^{(N)}$
and the inner products of the bi-orthogonal polynomials $\langle P_q, Q_q\rangle$ for $q=0,1,2,\ldots$
The bi-orthogonal polynomials themselves are not the subject of our current investigation,
therefore it is reasonable to follow the approach of \cite{borodin}
and derive expressions
for $\kappa^{(N)}$ and $\langle P_q, Q_q\rangle$ directly in terms of the moment matrix $M$ defined as
\bea\label{eq_mmgen}
M_{ij}=\langle z^i, z^j\rangle,~i,j\geq 0.
\eea
Let $(L,D,U)$ be the LDU-decomposition of $M$. That is $D$ is the diagonal matrix, $L$ and $U^T$ are the lower
triangular matrices with the diagonal entries equal to $1$ such that
\bea\label{eq_ldu}
M=LDU.
\eea 
Therefore, $L^{-1}M U^{-1}=D$. Re-writing this identity in components we find that
\bea\label{eq_ip1}
\langle P_{k},Q_{l}\rangle=D_{kk}\delta_{k,l},~k,l\geq 0,
\eea
where 
\bea\label{eq_bop1}
P_{k}(z)=\sum_{m=0}^{k} (\bar{L}^{-1})_{km}z^m,\nonumber\\
\\
\nonumber
Q_{k}(z)=\sum_{m=0}^{k} z^m(U^{-1})_{mk},
\eea
for $k\geq 0$. We see that $\{P_{q}, Q_q\}_{k\geq 0}$ is the set of holomorphic monic polynomials 
bi-orthogonal with respect to the weight $\omega(\cdot, \cdot \mid \lambda_1, \bar{\lambda}_1)$.
Comparing (\ref{eq_ip1}) with (\ref{eq_bop}) we find that
\bea\label{eq_ip22}
\langle P_k,Q_k\rangle=D_{kk}, ~k\geq 0.
\eea
Substituting (\ref{eq_ip22}) and (\ref{eq_bop1}) into the expression (\ref{eq_redk})
for the reduced kernel we also find that
\bea\label{eq_redker}
\kappa^{(N)} (\bar{z}, z\mid \lambda_1, \bar{\lambda}_1)=\sum_{i,j=0}^{N-1} z^i C^{(N-1)}_{ij} \bar{z}^j,
\eea
where 
\bea
C_{ij}^{(N)}=\sum_{k=0}^{N} (U^{-1})_{ik}\frac{1}{D_{kk}} (L^{-1})_{kj},~i,j\geq 0.
\eea
At least formally, the semi-infinite matrix $C^{(N)}$ converges to $M^{-1}$ as $N\rightarrow \infty$.
Perhaps less trivially, as a consequence of 
the 
Gram-Schmidt orthogonalisation procedure, it can be also characterised as the inverse of the $(N+1)\times (N+1)$ moment
matrix $(\langle z^i, z^j \rangle)_{0\leq i,j}\leq N$, see \cite{borodin}.

We  are not attempting to justify the above formal operations with semi-infinite matrices in general, but
in Section \ref{sec_finiteN} a justification will be given 
for the integration weight at hand.

Now the proof of Theorem \ref{thm_fn} has been reduced to the calculation of the LDU decomposition of
of the moment matrix $M$.\\
\\
{\bf Remark.} We see that the expression for the off-diagonal overlap $D_{12}^{(N,k)}$ is determinantal with the kernel
expressed as the $2\times 2$ determinant of a matrix built out of the kernel corresponding to the weight
$\omega(x,\bar{x}\mid\lambda, \bar{\lambda})$. Such a structure is to be
expected from the general theory of orthogonal polynomials in the complex plane developed in
\cite{akemann_vernizzi}. Really, relation (\ref{eqnd12int}) can be re-written as
\bea
D_{12}^{(N,k)}(\lv^{(k)})&=&-\frac{e^{-|\lambda_1|^2-|\lambda_2|^2}}{Z_N}
\frac{N!}{(N-k)!}
\int_{\C^{N-k}}
\prod_{i=k+1}^{N}d\lambda_i
d\bar{\lambda}_i
|\Delta^{(N-2)}(\lambda_3,\lambda_4,\ldots, \lambda_N)|^2 \nonumber\\
&&\times\prod_{m=3}^N(\lambda_2-\lambda_m)(\bar{\lambda}_1-\bar{\lambda}_m)\pi \omega(\lambda_m, \bar{\lambda}_m\mid \lambda_1, \bar{\lambda}_2).
\nonumber
\eea
By Dyson's theorem, the right hand side of this expression is proportional to the $(k-2)\times (k-2)$ determinant
of the kernel associated with holomorphic polynomials, which are bi-orthogonal with respect to the weight
\[
(u-z)(\bar{v}-\bar{z})\omega(z, \bar{z}\mid v, \bar{u}).
\]
Such a kernel can be expressed in terms of a $2\times 2$ determinant of the kernel associated
with the weight $\omega(\cdot, \cdot \mid \lambda,\bar{\lambda})$, see formula (3.10)
of \cite{akemann_vernizzi}, 
which can be considered as a generalisation of Christoffel's theorem for orthogonal polynomials in the complex plane.
Our present calculation can be therefore regarded as a short re-derivation of the general expression of 
\cite{akemann_vernizzi}
in the particular context of integration weights associated with the overlaps. The main tools used in our
calculation are the analyticity and 
determinant identities.\\
\\

\subsection{Lemma \ref{thm_rel}}\label{proof_of_lemma_1}\label{sec_La1}
It follows from 
(\ref{eqnd11int}) and (\ref{eqnd12int}) 
that both $D_{11}^{(N,k)}(\lv^{(k)})e^{\sum_{m=1}^k\lambda_k \bar{\lambda}_k}$
and $D_{12}^{(N,k)}(\lv^{(k)})e^{\sum_{m=1}^k\lambda_k \bar{\lambda}_k}$ 
are polynomials in $\lv^{(k)},\bar{\lv}^{(k)}$. 
Therefore $D^{(N,k)}_{11}$ and $D_{12}^{(N,k)}$ are entire functions on $\C^{2k}$.

Recall the definition of the transposition $\hat{T}$ acting on functions on $\C^{2k}$:
\bea
\hat{T}f(\lambda_1, \bar{\lambda}_1, \lambda_2, \bar{\lambda}_2,\ldots)
=f(\lambda_1, \bar{\lambda}_2, \lambda_2, \bar{\lambda}_1,\ldots)
\eea
Comparing (\ref{eqnd11int}) and (\ref{eqnd12int}), we see that for $k\geq 2$,
\bea\label{rltn1}
D_{12}^{(N,k)}(\lv^{(k)})=- \frac{e^{-|\lambda_1-\lambda_2|^2}}{1-|\lambda_1-\lambda_2|^2}\hat{T}D_{11}^{(N,k)}(\lv^{(k)}),  
\eea
for any $(\lv^{(k)},\bar{\lv}^{(k)}) \in \C^{2k}$. Lemma 1 is proved.\\
\subsection{Heuristic derivation of $N= \infty$ results in the bulk  assuming T-invariance}\label{sec_heur}
The task of calculating bi-orthogonal polynomials (\ref{eq_bop}) is considerably simpler at the
special point 
$\lambda_1=\bar{\lambda}_1=0$. In this case the weight function 
reduces to
\bea
\omega(z,\bar{z}\mid 0,0)=\frac{1}{\pi}(1+|z|^2)e^{-|z|^2},
\eea
which is an $U(1)$-invariant function. The bi-orthogonal polynomials associated with $U(1)$-invariant weights are just the monomials,
\bea
P_{k}(x)=Q_{k}(x)=x^k, ~k\geq 0.
\eea
Their inner products can also be computed explicitly,
\bea\label{eq_monip}
\langle P_k ,Q_k\rangle= k! (k+2),~k\geq 0,
\eea
leading to the following kernel:
\bea
K^{(N)}_{11}(x,\bar{x},y,\bar{y}\mid 0,0)=\frac{1}{\pi}(1+|x|^2)e^{-|x|^2}\sum_{k=0}^{N-1}
\frac{(\bar{x}y)^k}{(k+2)k!}.
\eea
As $N\rightarrow \infty$ , the limiting kernel in the bulk is
\bea
K^{(bulk)}_{11}(x,\bar{x},y,\bar{y}\mid 0,0)&=&\frac{1}{\pi}(1+|x|^2)e^{-|x|^2} \kappa^{(bulk)}(\bar{x},y\mid 0,0),
\eea
where
\bea
 \kappa^{(bulk)}(\bar{x},y\mid 0,0)=
\frac{1}{(\bar{x}y)^2}+\left(\frac{1}{(\bar{x}y)}-\frac{1}{(\bar{x}y)^2}\right)e^{\bar{x}y}.
\eea
Alternatively, we can write
\bea
\kappa^{(bulk)}(\bar{x},y\mid 0,0)=
\left.\frac{d}{dz}\left(\frac{e^z-1}{z}\right)
\right|_{z=\bar{x}y}.
\eea
The $N$-dependent pre-factor in the right hand side of (\ref{eqd11_gen}) is $N/\pi$, which leads to 
the following answer for the conditional overlap in the bulk:
\bea
D_{11}^{(bulk,~k)}(0,\lambda_2, \ldots, \lambda_k)=
\frac{1}{\pi} \det_{2\leq i,j\leq k}
\left(K^{(bulk)}_{11}(\lambda_i,\bar{\lambda}_i,\lambda_j,\bar{\lambda}_j\mid 0,0)
\right)
\eea

Let us $assume$ the extended translational invariance for the diagonal
overlaps regarded as functions on $C^{2k}$, which means that 
$D_{11}^{(bulk,~k)}$ is invariant under the shift
$\lambda_m \rightarrow \lambda_m+\epsilon$, $\bar{\lambda}_m\rightarrow \bar{\lambda}_m+\bar{\epsilon}$,
$m=1,2,\ldots , k$, where $\epsilon,\bar{\epsilon}$ are independent complex variables. Then
\bea
D_{11}^{(bulk,~k)}(\lv^{(k)})&=&D_{11}^{(bulk,~k)}(0,\lambda_2-\lambda_1, \ldots, \lambda_k-\lambda_1)\nonumber
\\
&=&\frac{1}{\pi} \det_{2\leq i,j\leq k}
\left(K^{(bulk)}_{11}(\lambda_i-\lambda_1,
\bar{\lambda}_i-\bar{\lambda}_1,
\lambda_j-\lambda_1,\bar{\lambda}_j-\bar{\lambda}_1\mid 0,0)\right).
\eea
We conclude that 
\bea\label{eqd11_heur}
D_{11}^{(bulk,~k)}(\lv^{(k)})=\frac{1}{\pi} \det_{2\leq i,j\leq k}
\left(K^{(bulk)}_{11}(\lambda_i,
\bar{\lambda}_i,
\lambda_j,\bar{\lambda}_j\mid \lambda_1,\bar{\lambda}_1)\right),
\eea
where
\bea
K^{(bulk)}_{11}(x,\bar{x},y,\bar{y}\mid \lambda,\bar{\lambda})&=&\frac{1}{\pi}(1+|x-\lambda|^2)e^{-|x-\lambda|^2} \kappa^{(bulk)}(\bar{x},y\mid \lambda,\bar{\lambda}),
\eea
and the reduced kernel is
\bea
 \kappa^{(bulk)}(\bar{x},y\mid \lambda,\bar{\lambda})=
\left.\frac{d}{dz}\left(\frac{e^z-1}{z}\right)
\right|_{z=(\bar{x}-\bar{\lambda})(y-\lambda)},
\eea
which agrees 
with the statement (\ref{ovlp11}) of Corollary \ref{thm_bulk}.

To calculate the off-diagonal conditional overlaps, let us $assume$
that the relation (\ref{rltn}) remains valid at $N=\infty$ as well. Then
\bea
D_{12}^{(bulk,~k)}(\lv^{(k)})=-
\frac{1}{\pi}
\frac{e^{-|\lambda_1-\lambda_2|^2}}{1-|\lambda_1-\lambda_2|^2}
\hat{T}\det_{2\leq i,j\leq k}
\left(K^{(bulk)}_{11}(\lambda_i,
\bar{\lambda}_i,
\lambda_j,\bar{\lambda}_j\mid \lambda_1, 
\bar{\lambda}_1)\right).
\eea
Applying 
the determinant 
identity (\ref{eq_ti}), we find
\bea\label{eqd12_heur}
D_{12}^{(bulk,~k)}(\lv^{(k)})=-\frac{1}{\pi^2} 
\kappa^{(bulk)}(\bar{\lambda}_1, \lambda_2\mid \lambda_1, \bar{\lambda}_2)
\det_{3\leq i,j\leq k}
\left(K^{(bulk)}_{12}(\lambda_i,
\bar{\lambda}_i,
\lambda_j,\bar{\lambda}_j\mid \lambda_1,\bar{\lambda}_2)\right),
\eea
where
\bea\label{eq_ker12}
K^{(bulk)}_{12}(\lambda_i,
\bar{\lambda}_i,
\lambda_j,\bar{\lambda}_j\mid \lambda_1,\bar{\lambda}_2)&=&
\frac{\omega^{(bulk)}(\lambda_i, \bar{\lambda}_i\mid \lambda_1, \bar{\lambda}_2)}
{\kappa^{(bulk)}(\bar{\lambda}_1, \lambda_2\mid \lambda_1, \bar{\lambda}_2)}\\
&&\times\det
\left(\begin{array}{cc}
\kappa^{(bulk)}(\bar{\lambda}_1, \lambda_2\mid \lambda_1, \bar{\lambda}_2)&
\kappa^{(bulk)}(\bar{\lambda}_1, \lambda_j\mid \lambda_1, \bar{\lambda}_2)\\
\kappa^{(bulk)}(\bar{\lambda}_i, \lambda_2\mid \lambda_1, \bar{\lambda}_2)&
\kappa^{(bulk)}(\bar{\lambda}_i, \lambda_j\mid \lambda_1, \bar{\lambda}_2)
\end{array}\right),\nonumber
\eea
which agrees with the statement (\ref{ovlp12}) of Corollary \ref{thm_bulk}.

The above calculation is rather simple, but non-rigorous - it rests on the assumptions
of the extended translational invariance of conditional overlaps in the bulk and the validity
of Lemma \ref{thm_rel} at $N=\infty$. We could try justifying these assumptions using 
analysis, but as it turns out, it is possible to obtain a fairly simple explicit expression
for the kernel at $N<\infty$, thus enabling the study of conditional overlaps not only in the bulk
of the spectrum, but also near the spectral edge. Notice that in the latter case the translational invariance is absent in principle.

\subsection{The kernel for $N<\infty$}\label{sec_finiteN}
\subsubsection{The LDU decomposition of the moment matrix.}
We will use the relation between the kernel and the moment matrix established
in Section \ref{sec_setup}. An explicit computation of $\langle z^i, z^j\rangle$ with the weight
$\omega(\cdot,\cdot\mid, \lambda, \bar{\lambda})$ defined (\ref{eqnwt}) gives
\bea
M_{ij}= i!\left[\delta_{ij}\left((1+\lambda \bar{\lambda})+(i+1) \right)- \delta_{i+1,j} \lambda(i+1)
- \delta_{i,j+1}\bar{\lambda} 
\right],~i,j\geq 0.
\eea
Crucially, the moment matrix is tri-diagonal, which makes explicit calculations leading to the kernel
possible. The recursive formulae for computing the LDU decomposition and the inverse of a tri-diagonal
matrix are well-known. What makes our case special however, is that the recursions 
we get can be solved exactly in terms of the exponential polynomials. At some point it would be interesting
to understand the algebraic reasons for the exact solvability of our problem, but in the mean time we adopt
 a tour de force
approach.

Let $\mu$ be the following tri-dagonal matrix:
\bea
\mu_{ij}=\delta_{ij}\left((1+\lambda \bar{\lambda})+(i+1) \right)- \delta_{i+1,j} \lambda(i+1)
- \delta_{i,j+1}\bar{\lambda}
,~i,j\geq 0.
\eea
As $M$ is the product of of $\mu$ and the diagonal matrix with entries 
$i!$, the LDU decomposition
of $M$ is easy to construct from the LDU decomposition of $\mu$. If
\bea
\mu= LDU,
\eea 
where $D_{pq}=d_p \delta_{pq}, L_{pq}=\delta_{pq}+l_p \delta_{p,q+1}, U_{pq}=\delta_{pq}+u_{q}\delta_{q,p+1}$,
$p,q \geq 0$, then
\bea
u_{p+1}&=&-\frac{(1+p) \lambda}{d_p},~p\geq 0,\\
l_{p+1}&=&-\frac{\bar{\lambda}}{d_p}, ~p\geq 0,\\
d_p&=&-d_{p-1} l_p u_{p} \mathbbm{1}_{p \geq 1}+2+\lambda \bar{\lambda} +p,~p\geq 0, 
\eea
defining $d_{-1}\equiv0$.
Let $x=\lambda \bar{\lambda}$. To determine the LDU decomposition of $\mu$ we have to solve the
first order non-linear recursion for $d_p$'s:
\bea
d_p=2+x+p-\frac{px}{d_{p-1}},~p\geq 1,\nonumber\\
\label{eq_nlrec}\\
d_0=2+x.\nonumber
\eea
This recursion can be linearised via the substitution $d_p=\frac{r_{p+1}}{r_p}$, which, 
upon choosing $r_0=1$,
gives
\bea
r_{p+1}+px r_{p-1}=(2+x+p)r_p,~p\geq 1,\nonumber\\
\label{eq_lrec}\\
r_1=2+x.\nonumber
\eea
The unique solution of (\ref{eq_lrec}) is
\bea
r_p(x)=p! \sum_{m=0}^p \frac{(p+1-m)}{m!}x^m=(p+1)!e_{p}(x)-p! x e_{p-1}(x),
\eea
where $e_p(x)=\sum_{k=0}^p\frac{x^p}{p!}$ is the exponential polynomial of degree $p$.

Therefore,
\bea
r_p(x)=p!f_{p}(x),~p=0,1,\ldots,
\eea
where $f_p$'s are the polynomials defined in (\ref{fpols}).
Converting the LDU decomposition of $\mu$ to the LDU decomposition of $M$ and updating notations, we
find that $M=LDU$, where
\bea\label{eq_ldul}
L_{pm}=\delta_{pm}-\bar{\lambda}\frac{f_{p-1}(\lambda \bar{\lambda})}{f_p(\lambda\bar{\lambda})} \delta_{p,m+1},~p,m \geq 0,\\
\label{eq_ldud}
D_{mm}= (m+1)! \frac{f_{m+1}(\lambda \bar{\lambda})}{f_m(\lambda \bar{\lambda})},~m\geq 0,\\
 U_{mq}=\delta_{mq}-\lambda 
 \frac{f_{m-1}(\lambda \bar{\lambda})}{f_m(\lambda\bar{\lambda})} \delta_{q,m+1},~m,q \geq 0.\label{eq_lduu}
\eea
Using the relation (\ref{eq_ip1}) between the LDU decomposition and the inner products of the bi-orthogonal
polynomials, we conclude that
\bea\label{eq_ip2}
\langle P_p, Q_p \rangle =  
(p+1)! \frac{(p+2)e_{p+1}(x)-
x e_{p}(x)}{(p+1)e_{p}(x)-x e_{p-1}(x)},
\eea
which coincides with (\ref{eq_monip}) at the point $x=0$, as it should. 
\subsubsection{The inner products of
the bi-orthogonal polynomials and the pre-factor in  (\ref{eqd11_gen}) }
Now we can calculate the factor in front of the determinant in the r. h. s.  of
(\ref{eqd11_gen}). Using the relation (\ref{eq_ip2}) we find
\bea
\prod_{q=0}^{N-2}
\langle P_q ,Q_q \rangle =
\prod_{q=0}^{N-2}
D_{qq} = \prod_{q=1}^{N-1}q! \cdot f_{N-1}(\lambda \bar{\lambda}).
\eea
Therefore, 
\bea
\pi^{N-1}\frac{N!}{Z_N} \prod_{q=0}^{N-2}\langle P_q ,Q_q \rangle \cdot e^{-\lambda \bar{\lambda}}
=\frac{f_{N-1}(\lambda \bar{\lambda})}{\pi }e^{-\lambda \bar{\lambda}},
\eea
which allows us to make the operation of $\hat{T}$ on the inner product explicit.

\subsubsection{Inversion of the $L$ and $U$ factors, the bi-orthogonal polynomials and the kernel}
The inverse of the lower-triangular matrix $L$ (resp. upper triangular matrix $U$) is a lower
(resp. upper) triangular matrix. The corresponding matrix elements can be computed
directly from the relations $LL^{-1}=I$, $U U^{-1}=I$ using the explicit expressions
(\ref{eq_ldul}) for the decomposition factors.  The answer is
\bea\label{eq_luinv}
(L^{-1})_{pq}=\left\{
\begin{array}{cc}
0& q>p,\\
1& q=p,\\
\bar{\lambda}^{p-q} \frac{f_q(x)}{f_p(x)}&q<p,
\end{array}
\right.
~~~
(U^{-1})_{pq}=\left\{
\begin{array}{cc}
\lambda^{q-p} \frac{f_p(x)}{f_q(x)}& q>p,\\
1& q=p,\\
0&q<p.
\end{array}
\right.
\eea
Substituting (\ref{eq_luinv}) and (\ref{eq_ldud}) into the formula (\ref{eq_redker}) 
we find that
\bea\label{eq_redker1}
\kappa^{(N+1)}(\bar{\mu}, \nu\mid \lambda,\bar{\lambda})=G^{(N)} 
\left(\lambda\bar{\lambda},\frac{\bar{\mu}}{\bar{\lambda}},\frac{\nu}{\lambda}\right),
\eea
where
\bea
G^{(N)}(x,y,z)
= \sum_{m,n=0}^N f_m(x) f_n(x) y^m 
z^n \sum_{k=m\vee n}^N \frac{x^k}{(k+1)! f_{k}(x)f_{k+1}(x)}
\eea
is a function on $\C^3$ and $a\vee b:=\max(a,b)$.
Even though we do not use explicit expressions for the bi-orthogonal polynomials in the paper, we 
record them here for future use:
\begin{eqnarray}
P_{k}(z\mid \lambda, \lambdab)=Q_{k}(z\mid \lambda, \lambdab)=\sum_{m=0}^k 
\lambda^{k-m}\frac{f_{m}(\lambda\lambdab)}{f_k(\lambda \lambdab)}z^m.
\end{eqnarray}
\subsubsection{Simplification of the reduced kernel for $N<\infty$.}
The above form of the reduced kernel is not well suited  for studying the large-$N$ asymptotic
of the overlaps. In particular, we do not see how to calculate the large-$N$ limit of the kernel
directly from (\ref{eq_redker1}). Fortunately, it can be considerably simplified
via a sequence of lucky cancellations yielding formula (\ref{thm1_redkern}).

The inner sum in (\ref{eq_redker1}) can be simplified as follows: 
Let $\Phi_n: \C\rightarrow \C$ be such that 
\bea
\Phi_n(x):=\sum_{k=0}^n \frac{x^k}{(k+1)!f_{k}(x)f_{k+1}(x)},
\eea
where we define $\Phi_{-1}\equiv 0$.
Then
\bea\label{kern1}
G^{(N)}(x,y,z)
= \sum_{m,n=0}^N f_m(x) f_n(x) y^m 
z^n \left(\Phi_{N}(x)-\Phi_{(m\vee n)-1}(x)\right).
\eea
We have the following key technical result:
\begin{lemma}\label{thm_step1}
\bea\label{phin}
\Phi_{n}(x)=\frac{(n+2-x)}{x^2f_{n+1}(x)}+ \frac{x-1}{x^2},~n=0,1,\ldots
\eea
\end{lemma}
\begin{proof}

For a fixed value of $x$, the sequence 
\bea
\Phi_n(x)=\sum_{k=0}^n \frac{x^k}{(k+1)!f_{k}(x)f_{k+1}(x)}
\eea
satisfies the following difference equation:
\bea
\Phi_{n+1}(x)&=&\Phi_n(x)+\frac{x^{n+1}}{(n+2)!f_{n+1}(x)f_{n+2}(x)},\label{diff_eqn}\\
\Phi_0(x)&=&\frac{1}{2+x}.\label{icond}
\eea
Using 
$f_1(x)=2+x$, it is easy to check that the expression (\ref{phin}) satisfies the initial condition (\ref{icond}).
Assuming that $\Phi_n$ is given by (\ref{phin}), we find from the equation (\ref{diff_eqn}) that
\bea
\Phi_{n+1}=\frac{x-1}{x^2}+\frac{(n+2-x)f_{n+2}(x)+\frac{x^{n+3}}{(n+2)!}}{x^2f_{n+1}(x)f_{n+2}(x)}.
\eea
A direct calculation based on the definitions (\ref{expols}) for the exponential polynomials $e_n$
and
 (\ref{fpols}) for the polynomials $f_n$
 confirms that
 \bea
 (n+2-x)f_{n+2}(x)+\frac{x^{n+3}}{(n+2)!}=(n+3-x)f_{n+1}(x).
 \eea
 Therefore,
 \bea
 \Phi_{n+1}=\frac{x-1}{x^2}+\frac{(n+3-x)}{x^2f_{n+2}(x)},
 \eea
 and Lemma \ref{thm_step1} is proved by induction.
\end{proof}
Substituting (\ref{phin}) into (\ref{kern1}), we find
\bea\label{kern2}
G^{(N)}(x,y,z)&=&\frac{(N+2-x)}{x^2 f_{N+1}(x)}
\left(\sum_{m=0}^N f_{m}(x)y^m\right)
\left(\sum_{n=0}^N f_{n}(x)z^n \right)\nonumber\\
&&+ \frac{1}{x^2}\sum_{m,n=0}^N (x-(m \vee n)-1)\frac{f_m(x) f_n(x)}{f_{m \vee n}(x)} y^m z^n.
\eea
Let 
\bea
\alpha_{n}(x,y):=\sum_{m=0}^n f_{m}(x)y^m, m=0,1,\ldots.
\eea
Then, the first term in the r.h.s.  of (\ref{kern2}) is equal to 
\bea\label{comp1}
\frac{(N+2-x)}{x^2 f_{N+1}(x)}\alpha_N(x,y)\alpha_N(x,z).
\eea
The second term can be also be expressed in terms of $\alpha_n$'s: 
\begin{eqnarray*}
&& \frac{1}{x^2}\sum_{m,n=0}^N (x-(m \vee n)-1)\frac{f_m(x) f_n(x)}{f_{m \vee n}(x)} y^m z^n=
\frac{1}{x^2}\sum_{n=0}^N (x- n-1)f_n(x)  (y z)^n
\\
&&+\frac{1}{x^2}\sum_{m>n\geq 0}^N (x-m -1) f_n(x) y^m z^n
+\frac{1}{x^2}\sum_{0\leq m<n}^N (x-n-1)f_m(x)  y^m z^n
\end{eqnarray*}
\bea\label{comp2}
=\frac{1}{x^2}\left(\left(x-\omega \frac{\partial}{\partial \omega}-1\right)\alpha_N(x,\omega)\mid_{\omega=yz}
+\psi_x(y,z)+\psi_x(z,y)\right),
\end{eqnarray}
where
\bea
\psi_x(y,z)=\sum_{n> m\geq 0}^N (x-n-1) f_m(x)y^m z^n.
\eea
Next,
\begin{eqnarray}
\psi_x(y,z)&=&\left(x-z\frac{\partial}{\partial z}-1  \right)
\sum_{n> m\geq 0}^N f_m(x)y^m z^n\nonumber\\
&=&\left(x-z\frac{\partial}{\partial z}-1  \right)\sum_{ m= 0}^N  f_m(x)y^m \sum_{n=m+1}^N z^n\nonumber\\
&=&\left(x-z\frac{\partial}{\partial z}-1  \right)\sum_{ m= 0}^N  f_m(x)y^m \left(\frac{z^{N+1}-z^{m+1}}{z-1}\right)\nonumber\\
\label{comp3}
&=&\left(x-z\frac{\partial}{\partial z}-1  \right)\frac{1}{(z-1)} \left( z^{N+1}\alpha_N(x,y)-z
\alpha_N(x,yz)\right) .
\end{eqnarray}
Substituting (\ref{comp3}) into (\ref{comp2}) and then substituting the result and (\ref{comp1}) into 
(\ref{kern2}), we find that
\bea
G^{(N)}(x,y,z)
&=&\frac{(N+2-x)}{x^2f_{N+1}(x)}\alpha_N(x,y)\alpha_N(x,z)
+\frac{1}{x^2}\left(x-\omega \frac{\partial}{\partial \omega}-1\right)
\alpha_N(x,\omega)\mid_{\omega=yz}
\nonumber\\
&&+\frac{1}{x^2}\left(x-z \frac{\partial}{\partial z}-1\right)\frac{z}{z-1}\left(z^N\alpha_N(x,y)-\alpha_N(x,yz)\right)
\nonumber\\
\label{kern3}
&&+\frac{1}{x^2}\left(x-y \frac{\partial}{\partial y}-1\right)\frac{y}{y-1}\left(y^N\alpha_N(x,z)-\alpha_N(x,yz)\right).
\eea
To simplify the expression for $G^{(N)}$ further,
we need an expression for $\alpha_N$ in terms of the exponential polynomials:
\bea
\alpha_N(x,y)&
=&\sum_{n=0}^N f_n(x)y^n=\sum_{n=0}^N\left((n+1)e_{n}(x)-xe_{n-1}(x) \right)
y^n
\nonumber\\
&=&\left(\frac{\partial}{\partial y}y-xy \right)\sum_{n=0}^N e_n(x)y^n+xy^{N+1} e_N(x)
\nonumber\\
&=&\left(\frac{\partial}{\partial y}y-xy \right)\frac{e_N(yx)-y^{N+1}e_N(x)}{1-y}+xy^{N+1} e_N(x).
\eea
Explicitly,
\bea\label{alphan}
\alpha_N(x,y)=
\frac{e_N(yx)}{(1-y)^2}-\frac{yx}{(1-y)}\frac{(yx)^N}{N!}-
((N+2-x)-(N+1-x)y)\frac{y^{N+1}e_N(x)}{(1-y)^2}.
\eea
Substituting (\ref{alphan}) into (\ref{kern3}), computing the derivatives and
grouping the terms according to the denominators we arrive at
\bea\label{kern001}
x^2G^{(N)}(x,y,z)=\frac{T^{(N)}_A(x,y,z)}{(1-y)^2(1-z)^2}
+\frac{T^{(N)}_B(x,y,z)}{(1-y)(1-z)}
+\frac{T^{(N)}_C(x,y,z)}{(1-y)^2(1-z)}
+\frac{T^{(N)}_C(x,z,y)}{(1-z)^2(1-y)},\quad
\eea
where
\bea\label{kern001Ta}
T^{(N)}_A(x,y,z)&=&\frac{(N+2-x)}{f_{N+1}(x)}e_{N}(yx)e_{N}(zx)-e_N(xyz)\\
&&+\frac{(N+2)}{f_{N+1}(x)(N+1)!}
\left((zx)^{N+1}e_{N}(yx)+(yx)^{N+1}e_N(zx)-(xyz)^{N+1}e_N(x)\right)
,
\nonumber
\eea
\bea\label{kern001Tb}
T^{(N)}_B(x,y,z)&=&\frac{(N+2-x)}{f_{N+1}(x)}\frac{(yx)^{N+1}}{N!}\frac{(zx)^{N+1}}{N!}
+xe_N(xyz)+(N+1-x)\frac{(xyz)^{N+1}}{N!}
\\
&&
-\frac{(N+2)(N+1-x)}{f_{N+1}(x)(N+1)!}\nonumber\\
&&\times\left(\frac{(zx)^{N+1}(yx)^{N+1}}{N!}
+\frac{(yx)^{N+1}(zx)^{N+1}}{N!}+(N+1-x)(xyz)^{N+1}e_N(x)\right)
,
\nonumber
\eea
\bea\label{kern001Tc}
&&
T^{(N)}_C(x,y,z)=-\frac{(N+2-x)}{f_{N+1}(x)}e_{N}(yx)\frac{(zx)^{N+1}}{N!}+\frac{(xyz)^{N+1}}{N!}
+\frac{(N+2)}{f_{N+1}(x)(N+1)!}
\\
&&\quad\quad\times
\left((N+1-x)(zx)^{N+1}e_N(yx)-\frac{(yx)^{N+1}(zx)^{N+1}}{N!}
-(N+1-x)e_N(x)(xyz)^{N+1}
\right)
.
\nonumber
\eea
A straightforward simplification of each of the $T$-terms gives:
\bea\label{kern002Ta}
f_{N+1}(x)T^{(N)}_A(x,y,z)&=&(N+2)(e_{N+1}(yx)e_{N+1}(zx)-e_{N+1}(xyz)e_{N+1}(x))\\
&&-x(e_{N}(yx)e_{N}(zx)-e_{N}(xyz)e_{N}(x))
\nonumber
,
\eea
\bea\label{kern002Tb}
f_{N+1}(x)T^{(N)}_B(x,y,z)&=&x\frac{(yx)^{N+1}(zx)^{N+1}}{N!(N+1)!}
+x(N+1-x)e_N(x)\frac{(xyz)^{N+1}}{(N+1)!}\\
&&+xe_N(xyz)\left((N+2-x)e_N(x)+(N+2)\frac{x^{N+1}}{(N+1)!}\right)
\nonumber
,
\eea
\bea\label{kern002Tc}
f_{N+1}(x)T^{(N)}_C(x,y,z)=x\frac{(xyz)^{N+1}}{(N+1)!}e_N(x)-xe_{N}(yx)\frac{(zx)^{N+1}}{(N+1)!}.
\eea
Substituting (\ref{kern002Ta}),  (\ref{kern002Tb}) and (\ref{kern002Tc}) into (\ref{kern001})
we arrive at
\bea\label{kern4}
x^2 f_{N+1}(x)G^{(N)}(x,y,z)&=&\frac{(N+2)W_{N+1}(x,y,z)-xW_{N}(x,y,z)}{(1-y)^2(1-z)^2}
\nonumber\\
&&+x\frac{(xyz)^{N+1}e_{N}(x)-(xz)^{N+1}e_{N}(xy)}{(N+1)!(1-y)^2(1-z)}
\nonumber\\
&&+x\frac{(xyz)^{N+1}e_{N}(x)-(xy)^{N+1}e_{N}(xz)}{(N+1)!(1-z)^2(1-y)}
\nonumber\\
&&-x\frac{(xyz)^{N+1}}{(N+1)!}\frac{e_{N+1}(x)+xe_{N}(x)}{(1-y)(1-z)},
\eea
where
\bea
W_N(x,y,z):=e_N(xy)e_N(xz)-(1-x(1-y)(1-z))e_N(xyz)e_N(x),~ N\in \mathbb{N}. 
\eea
This answer is already well-suited for the calculation of the kernel for the bulk and the edge
scaling limits of the overlaps, but it still looks rather complicated. Fortunately, it can be re-written in
a shorter form: 
Observing that $e_{n}(x)=e_{n+1}(x)-\frac{x^{n+1}}{(n+1)!}$, we find that
\bea\label{temp33}
W_{N}(x,y,z)&=&W_{N+1}(x,y,z)-e_{N+1}(xy)\frac{(xz)^{N+1}}{(N+1)!}
-e_{N+1}(xz)\frac{(xy)^{N+1}}{(N+1)!}\nonumber\\
&&+\left(e_{N+1}(xyz)\frac{(x)^{N+1}}{(N+1)!}
+e_{N+1}(x)\frac{(xyz)^{N+1}}{(N+1)!}\right)\left(1-x(1-y)(1-z)\right)\nonumber\\
&&+x(1-y)(1-z)\frac{(yzx^2)^{N+1}}{((N+1)!)^2}.
\eea
Expressing $W_{N+1}$ and $W_{N}$ in (\ref{kern4}) in terms of $W_{N+2}$ and $W_{N+1}$, 
using (\ref{temp33}) one finds
\bea\label{kern5}
x^2 f_{N+1}(x)G^{(N)}(x,y,z)&=&\frac{(N+2)W_{N+2}(x,y,z)-xW_{N+1}(x,y,z)}{(1-y)^2(1-z)^2}
\nonumber\\
&&- \frac{x}{(1-y)(1-z)}\frac{(xyz)^{N+2}}{(N+1)!}e_{N+2}(x).
\eea
It is straightforward to check that
\bea\label{temp34}
- x(1-y)(1-z)\frac{(xyz)^{N+2}}{(N+1)!}e_{N+2}(x)=(N+2)H_{N+2}(x,y,z)-xH_{N+1}(x,y,z),
\eea
where
\bea\label{temp35}
H_N(x,y,z):=\frac{(1-y)(1-z)}{N!} \frac{x^{N+1}e_N(xyz)-(xyz)^{N+1} e_N(x)}{(yz-1)}.
\eea
It follows from (\ref{kern5}, \ref{temp34}), that 
\bea\label{kern6}
x^2 f_{N+1}(x)G^{(N)}(x,y,z)
=\frac{(N+2)\frak{F}_{N+2}(x,y,z)
-x\frak{F}_{N+1}(x,y,z)}{(1-y)^2(1-z)^2},
\eea
where
\bea
\frak{F}_n(x,y,z):=W_{n}(x,y,z)+H_{n}(x,y,z), n ~\in \mathbb{N},
\eea
is a function on $\C^3$ defined in (\ref{total}).
Finally, substituting (\ref{kern6}) into (\ref{eq_redker1}) we arrive
at the expression (\ref{thm1_redkern}) for the reduced kernel. 
Theorem \ref{thm_fn} is proved.\\
\\
{\bf Remark.} The final part of the proof following (\ref{kern4}) is not very satisfying as it is both 
non-obvious and reliant on 
unexpected cancellations. A more direct route to (\ref{kern6}) is to
substitute the partition of unity $1=\ind_{m<n}+\ind_{m\geq n}$ into the double sum
in (\ref{kern1}), represent the indicator functions as a contour integral,
\[
\ind_{m<n}=\oint \frac{dz}{2\pi i } \frac{z^{m-n}}{1-z},
\]
and analyse the resulting expression for $G_N(x,y,z)$ as a sum of two contour integrals.
The integrands of each of the integrals contain poles of order $N$ as well 
as poles of order $1$ and $2$. It turns out that the contributions from the high order poles 
cancel, and the sum of contributions from the poles of low order gives (\ref{kern6}), see \cite{sakis} for details.
\subsection{The proof of Corollary \ref{thm_bulk}}\label{proof_bulk}
The proof is based on the following elementary remark: for any fixed $x \in \C$ 
\bea
\lim_{N\rightarrow \infty} e_N(x)=e^x.
\eea
Consequently,
\bea
\frac{f_{N}(x)}{N}=
(1+N^{-1})e_{N}(x)-N^{-1}x e_{N-1}(x)\stackrel{N\uparrow \infty}{\longrightarrow} e^{x}.
\eea
Therefore, the factor in front of the determinant in the r.h.s. of  (\ref{d11exact}) divided by
$N$ converges
for $|\lambda_1|^2<N$
 to
\bea\label{d11prefacbulk}
\lim_{N\rightarrow \infty}\frac{f_{N-1}(|\lambda_1|^2)}{N\pi }
e^{-|\lambda_1|^2}=\frac{1}{\pi},
\eea
as is well known from the Ginibre ensemble, cf. \cite{KS}. 
The large-$N$ limit of the reduced kernel defined in (\ref{thm1_redkern}) is 
most easily taken when fixing all arguments of the kernel, that is remaining in the vicinity of the origin, as the spectral edge is located at $\sqrt{N}$.
The same bulk limit close to  the origin was taken already in Ginibre's original paper
for the complex eigenvalue correlations \cite{ginibre}. For our kernel we thus have
\bea\label{redkerprime}
\lim_{N\rightarrow \infty} \kappa^{(N)}(\bar{x},y\mid \lambda,\lambdab)
&=&
e^{-\lambda \lambdab}  \frac{\lim_{N\rightarrow \infty}
\frak{F}_{N+1}\left(\lambda \lambdab,\frac{\bar{x}}{\lambdab}, \frac{y}{\lambda}\right)}
{(\bar{x}-\lambdab)^2(y-\lambda)^2}\nonumber
\\
&=& \frac{e^{\bar{x}y}}{(\bar{x}-\lambdab)^2(y-\lambda)^2}
\left( e^{-(\bar{x}-\bar{\lambda})(y-\lambda)}-1+(\bar{x}-\bar{\lambda})(y-\lambda)   \right)
\nonumber\\
&=& \frac{e^{\bar{x}y}}{(\bar{x}-\lambdab)^2(y-\lambda)^2}
 e^{(2)}(-(\bar{x}-\bar{\lambda})(y-\lambda)),
\eea
where $e^{(m)}(x):=\sum_{n=m}^\infty \frac{x^n}{n!},~m=0,1,\ldots$
In the second equality we used the definition (\ref{total}) of the polynomials
$\frak{F}_{N}$. Therefore, the bulk scaling limit of the kernel $K^{(N)}_{11}$ defined in (\ref{thm1_kern}) is
\bea\label{kernconj}
\lim_{N\rightarrow \infty} K^{(N)}_{11}(x,\bar{x},y,\bar{y}\mid \lambda,\bar{\lambda})
&=&\frac{e^{-x\bar{x}+y\bar{x}}}{\pi}\left(1+(x-\lambda)(\bar{x}-\lambdab)\right) 
\frac{ e^{(2)}(-(\bar{x}-\bar{\lambda})(y-\lambda))}{(\bar{x}-\lambdab)^2(y-\lambda)^2}
\nonumber\\
&=&e^{(y-x)\lambdab}
\frac{1}{\pi}(1+(x-\lambda)(\bar{x}-\lambdab))
e^{-(x-\lambda)(\bar{x}-\lambdab)}
\nonumber\\
&&\times 
\left.
\frac{d}{dz} \left(\frac{e^{z}-1}{z}\right)\right|_{z=(\bar{x}-\lambdab)(y-\lambda)}.
\eea
Notice that the factor $e^{(y-x)\lambdab}$ in the r.h.s. of (\ref{kernconj}) corresponds to 
the conjugation of the kernel $K(\lambda_{i},\lambda_j)\rightarrow \phi(\lambda_i) K(\lambda_i, \lambda_j)\phi^{-1}(\lambda_j)$, which does not change the value of the determinant in the expression (\ref{d11exact}) for the conditional overlap. This remark allows us to write the bulk scaling limit of the kernel as
\bea\label{kernunconj}
K^{(bulk)}_{11}(x,\bar{x},y,\bar{y}\mid \lambda,\bar{\lambda})
=
\frac{1}{\pi}(1+(x-\lambda)(\bar{x}-\lambdab))e^{-(x-\lambda)(\bar{x}-\lambdab)}
\left.
\frac{d}{dz} \left(\frac{e^{z}-1}{z}\right)\right|_{z=(\bar{x}-\lambdab)(y-\lambda)}.\nonumber\\
\eea
Expressions (\ref{d11prefacbulk}), (\ref{kernunconj}) solve the problem of computing the large-$N$
limit of (\ref{d11exact}), thus proving claims (\ref{ovlp11}), (\ref{thm_k11bulk}), (\ref{thm_weight11bulk})
and (\ref{thm_redker11bulk}) of Corollary \ref{thm_bulk}.

It remains to calculate the bulk scaling limit of $D_{12}^{(N,k)}$ starting with its determinantal 
representation (\ref{thm_d12exact}).  Using (\ref{redkerprime}) and (\ref{kernconj}), one finds
that the large-$N$ limit of the pre-factor in (\ref{thm_d12exact}) divided by $N$ is
\bea\label{d12bulkprefac}
\lim_{N\rightarrow \infty} \frac{e^{-|\lambda_1|^2-|\lambda_2|^2}}{N\pi^2}f_{N-1}(\lambda_1\bar{\lambda}_2)
\kappa^{(N-1)}(\bar{\lambda}_1,\lambda_2\mid \lambda_1,\bar{\lambda}_2)=
\frac{\kappa^{(bulk)} (\lambdab_1, \lambda_2\mid \lambda_1, \lambdab_2)}{\pi^2} ,
\eea
where $\kappa^{(bulk)}$ is defined in (\ref{thm_redker11bulk}), and
\bea\label{kernel12conj}
\lim_{N\rightarrow \infty} K^{(N-1)}_{12}(x,\bar{x},y,\bar{y}\mid u, \bar{u},v,\bar{v})
&=&
e^{(y-x)\bar{v}} 
\frac{\omega^{(bulk)}(x,\bar{x}\mid u,\bar{v})}{\kappa^{(bulk)}(\bar{u},v\mid u,\bar{v})} \nonumber\\
&&\times \det\left(
\begin{array}{cc}
\kappa^{(bulk)}(\bar{u}, v\mid u,\bar{v}) &\kappa^{(bulk)}(\bar{u}, y\mid u,\bar{v})\\
\kappa^{(bulk)}(\bar{x}, v\mid u,\bar{v}) & \kappa^{(bulk)}(\bar{x}, y\mid u,\bar{v})
\end{array}
\right),\quad\quad
\eea
where the weight $\omega^{(bulk)}$ is defined by (\ref{thm_weight11bulk}).
Calculating the large-$N$ limit of (\ref{thm_d12exact}) with the help of (\ref{d12bulkprefac})
and (\ref{kernel12conj}), and using the fact that conjugation of the kernel by $e^{\lambda_i \lambdab_2}$
does not change the determinant, we arrive at  the characterisation (\ref{ovlp12}), (\ref{eq_ker12bulk})
for the bulk scaling limit of the off-diagonal conditional overlaps. Corollary \ref{thm_bulk} is proved.
\subsection{Corollary \ref{thm_edge}}\label{proof_edge}
The calculation is based on the following two asymptotic formulae:
Let us fix $a,b \in \C$. Then
\bea\label{ase1}
\log \frac{(N+\sqrt{N}a+b)^{N+1}}{(N+1)!}&=&(N+\sqrt{N}a+b)-\frac{1}{2}\log2\pi N\nonumber\\
&&-\frac{a^2}{2}
+\frac{a-ab+\frac{1}{3}a^3}{\sqrt{N}}+O(N^{-1}),\\
e_{N+k}(N+\sqrt{N}a+b)&=&e^{N+\sqrt{N}a+b}\left(F(a)+\frac{e^{-\frac{a^2}{2}}}{\sqrt{2\pi N}}
\left(\frac{a^2}{3}-b+k+\frac{2}{3}\right)+O(N^{-1}) \right),\nonumber\\
\label{ase2}
\eea 
where $k=0,1,2,\ldots$. Here $F$ is a rescaling on 
the complementary error function defined in (\ref{eq_erfc}).
The derivation of the above formulae 
is based on Stirling's formula and
the following well known integral representation of 
the exponential polynomials
in terms of the incomplete Gamma-function:
\bea
e_{n}(x)
=e^x\frac{\Gamma(n+1,x)}{\Gamma(n+1)}=\frac{e^x}{n!}  \int_x^\infty t^{n} e^{-t} dt,~n=0,1,2\ldots,
\eea
see \cite[Chapter 8.11.10 ]{olver} for more details.
The calculations leading to 
(\ref{ase1}) and (\ref{ase2}) are straightforward, but lengthy due to the fact that
we need to know the asymptotic expansion of $e_{N+k}$ and $\log \frac{(N+\sqrt{N}a+b)^{N+1}}{(N+1)!}$
up to and including the terms of order $N^{-1/2}$, 
see also \cite{olver}.

As a consequence of (\ref{ase2}),
\bea\label{fedge}
f_{N+1}(N+\sqrt{N}a+b)&=&\sqrt{\frac{N}{2\pi}} e^{N+\sqrt{N}a+b-\frac{a^2}{2}} 
\nonumber\\
&&\times\left(1-a\sqrt{2\pi} e^{\frac{a^2}{2}}F(a)\right) (1+O(N^{-\frac{1}{2}})), a,b \in \C.\nonumber\\
\eea
Now we are ready to calculate the edge scaling limit of conditional overlaps.
Let 
\bea
x^{(N)}=e^{i\theta}(\sqrt{N}+x),\nonumber\\
y^{(N)}=e^{i\theta}(\sqrt{N}+y),\label{edge_coord}\\
\lambda^{(N)}=e^{i\theta}(\sqrt{N}+\lambda),\nonumber
\eea
where $x,y,\lambda \in \C$ 
are of order unity.
Then the edge scaling limit of the factor multiplying the determinant in the r.h.s. of (\ref{d11exact})
down-scaled by $N^{-1/2}$ 
is
\bea\label{prefac_edge11}
\lim_{N\rightarrow \infty} \frac{e^{-|\lambda^{(N)}|^2}}{\pi \sqrt{N}}f_{N-1}(|\lambda^{(N)}|^2)
=
\frac{1}{\sqrt{2\pi^3}}
\left(e^{-\frac{1}{2}\left(\lambda+\lambdab \right)^2}-\sqrt{2\pi}\left(\lambda+\lambdab\right)F(\lambda+\lambdab) \right), ~\lambda, \lambdab \in \C^2.\quad
\eea
The derivation of (\ref{prefac_edge11}) is based on (\ref{ase2}). 
Note that (\ref{prefac_edge11}) is valid for any pair of complex numbers $(\lambda,\lambdab)$, not
just on the real surface $\lambda=\bar{\lambdab}$, which makes it suitable for the calculation of 
both the diagonal and the off-diagonal overlaps. 

To find the edge scaling limit of the kernel $K_{11}^{(N)}$ we substitute the expressions (\ref{edge_coord})
into the formula for the kernel 
\begin{eqnarray*}
K_{11}^{(N)}(x^{(N)},\bar{x}^{(N)},y^{(N)},\bar{y}^{(N)}\mid \lambda^{(N)},\lambdab^{(N)})
&=&\frac{1+(x^{(N)}-\lambda^{(N)})(\bar{x}^{(N)}-\lambdab^{(N)})}{\pi} 
e^{-x^{(N)}\bar{x}^{(N)}}\\ 
&&\times
G^{(N-1)}\left( \lambda^{(N)} \lambdab^{(N)},
\frac{\bar{x}^{(N)}}{\lambdab^{(N)}}, \frac{y^{(N)}}{\lambda^{(N)}}\right),
\end{eqnarray*}
where $G^{(N)}$ is given by formula (\ref{kern5}), and compute the large-$N$ asymptotics using (\ref{ase1}), (\ref{ase2})
and (\ref{fedge}).  The result follows from another lengthy computation and is
\bea
K_{11}^{(N)}(x^{(N)},\bar{x}^{(N)},y^{(N)},\bar{y}^{(N)}\mid \lambda^{(N)},\lambdab^{(N)})
=e^{\sqrt{N}(y-x)}
\frac{1+(x-\lambda)(\bar{x}-\lambdab)}{\pi} 
e^{-x\bar{x}+\bar{x}y}\nonumber\\ 
\times
\frac{
H\left(\lambda+\lambdab,\lambda+\bar{x},\lambdab+y,y+\bar{x},(\lambda-y)(\lambdab-\bar{x})\right)}
{(\lambda-y)^2(\bar{\lambda}-\bar{x})^2}
\left(1+O(N^{-1/2})\right),
\eea
where the function $H$ is defined in (\ref{hedge}). We see that there the large-$N$ limit of $K_{11}^{(N)}$
at the edge does not exist, but fortunately, the residual $N$-dependence can be eliminated by the $N$-dependent
conjugation 
\bea\label{ndep_conj}
K(x,\bar{x}, y, \bar{y}\mid \lambda, \lambdab)\rightarrow e^{\sqrt{N}x}
K(x,\bar{x}, y, \bar{y}\mid \lambda, \lambdab) e^{-\sqrt{N}y},
\eea
which does not change the value of the conditional overlap. Therefore we can conclude that
\bea\label{kern33edge}
&&K_{11}^{(edge)}(x,\bar{x},y,\bar{y}\mid \lambda,\lambdab):=\lim_{N\rightarrow \infty} e^{\sqrt{N}x}
K_{11}^{(N)}(x^{(N)},\bar{x}^{(N)},y^{(N)},\bar{y}^{(N)}\mid \lambda^{(N)},\lambdab^{(N)})e^{-\sqrt{N}y}
\nonumber\\
&&=
\frac{1+(x-\lambda)(\bar{x}-\lambdab)}{\pi} 
e^{-x\bar{x}+\bar{x}y} 
\frac{
H\left(\lambda+\lambdab,\lambda+\bar{x},\lambdab+y,y+\bar{x},(\lambda-y)(\lambdab-\bar{x})\right)}
{(\lambda-y)^2(\bar{\lambda}-\bar{x})^2}.\nonumber\\
\eea
Substituting (\ref{prefac_edge11}) and (\ref{kern33edge}) into the edge scaling limit (\ref{def_d11edge})
of the conditional overlap $D_{11}^{(N,k)}$ we arrive at the statement 
(\ref{d11_edge}), (\ref{limkeredge}),
(\ref{limweightedge}) and  (\ref{limredkeredge}) of the Corollary \ref{thm_edge}.

To find the edge scaling limit of the off-diagonal overlap $D_{12}^{(N,k)}$, we need to 
substitute its expression (\ref{thm_d12exact}) into (\ref{def_d12edge}) and calculate the large-$N$
limit of the resulting sequence.  As before, the calculation reduces to the evaluation of the scaling
limits of the pre-factor in the r.h.s. of (\ref{thm_d12exact}) and the kernel (\ref{kern12}).

A straightforward computation based on (\ref{fedge}, \ref{kern33edge}) gives 
\bea\label{prefac12_edge}
&&\lim_{N\rightarrow \infty}
(-1)\frac{e^{-|\lambda^{(N)}_1|^2-|\lambda^{(N)}_2|^2}}{\pi^2\sqrt{N}}f_{N-1}(\lambda^{(N)}_1\lambdab^{(N)}_2)
\kappa^{(N-1)}(\lambdab^{(N)}_1,\lambda^{(N)}_2\mid \lambda^{(N)}_1,\lambdab^{(N)}_2)
\nonumber\\
&=&
-\frac{e^{-|\lambda_{12}|^2-\frac{1}{2}(\lambda_1+\lambdab_2)^2}}{\sqrt{2\pi^5}\lambda_{12}^2\lambdab_{12}^2}
\left(1-\sqrt{2\pi}(\lambda_1+\lambdab_2)e^{\frac{1}{2}(\lambda_1+\lambdab_2)^2}F(\lambda_1+\lambdab_2)\right)
\nonumber\\
&&\times H(\lambda_1+\lambdab_2,\lambda_1+\lambdab_1,\lambda_2+\lambdab_2,
\lambda_2+\lambdab_1,-\lambda_{12}\lambdab_{12}).
\eea
Similarly, introducing in addition to (\ref{edge_coord}), $u^{(N)}=e^{i\theta}(\sqrt{N}+u)$, 
$v^{(N)}=e^{i\theta}(\sqrt{N}+v)$, we find that
\bea\label{kern12_penult}
&&K_{12}^{(N)} (x^{(N)},\bar{x}^{(N)},y^{(N)},\bar{y}^{(N)}\mid u^{(N)},\bar{u}^{(N)},v^{(N)},\bar{v}^{(N)})
\nonumber\\
&=&
e^{\sqrt{N}(y-x)}K_{12}^{(edge)} (x,\bar{x},y,\bar{y}\mid u,\bar{u},v,\bar{v})
\left(1+O(N^{-1/2})\right),
\eea
where $K_{12}^{(edge)}$ is given by (\ref{eq_ker12edge}).
Substituting (\ref{prefac12_edge}) and (\ref{kern12_penult}) into the right hand side
of (\ref{def_d12edge}), performing the $N$-dependent conjugation (\ref{ndep_conj}) 
of the kernel inside the resulting determinant and computing the $N\rightarrow \infty$ limit,
we arrive at the statements (\ref{d12_edge}), (\ref{eq_ker12edge}) of the Corollary \ref{thm_edge}.
Corollary \ref{thm_edge} is proved.
\subsection{The proof of corollary \ref{thm_edge_bulk}}\label{proof_edge_to_bulk}
It follows from (\ref{ovlp11}), (\ref{ovlp12}), (\ref{d11_edge}) and (\ref{d12_edge})
that
\bea
D_{11}^{(bulk,~k)}(\lv^{(k)})&=&D_{11}^{(bulk,~1)}(\lambda_1)\det_{2\leq i,j\leq k}(K^{(bulk)}_{11}(\lambda_i,
\bar{\lambda}_i,
\lambda_j,\bar{\lambda}_j\mid \lambda_1,\bar{\lambda}_1))
\label{ovlp11prime},~k\geq 1,
\\
D_{12}^{(bulk,~k)}(\lv^{(k)})&=&D_{12}^{(bulk,~2)}(\lambda_1,\lambda_2) \det_{3\leq i,j\leq k}(K^{(bulk)}_{12}(\lambda_i,
\bar{\lambda}_i,
\lambda_j,\bar{\lambda}_j\mid \lambda_1,\lambdab_1,\lambda_2,\bar{\lambda}_2)),
\nonumber\\
~k\geq 2,
\label{ovlp12prime}
\\
D_{11}^{(edge,~k)}(\lv^{(k)})&=&D_{11}^{(edge,~1)}(\lambda_1)\det_{2\leq i,j\leq k}(K^{(edge)}_{11}(\lambda_i,
\bar{\lambda}_i,\lambda_j,\bar{\lambda}_j\mid \lambda_1,\bar{\lambda}_1)), ~k\geq 1,
\label{ovlp11edgeprime}
\\
D_{12}^{(edge,~k)}(\lv^{(k)})&=&D_{12}^{(edge,~2)}(\lambda_1,\lambda_2)\det_{3\leq i,j\leq k}(K^{(edge)}_{12}(\lambda_i,
\bar{\lambda}_i,
\lambda_j,\bar{\lambda}_j\mid \lambda_1,\lambdab_1,\lambda_2,\bar{\lambda}_2)),\nonumber
\\
~k\geq 2,
\label{ovlp12edgeprime}
\eea
where the expressions for all the relevant kernels are given in Corollaries \ref{thm_bulk} and \ref{thm_edge}.
We see that the ratios 
\bea
\frac{D_{11}^{(edge,~k)}}{D_{11}^{(edge,~1)}}, \frac{D_{12}^{(edge,~k)}}{D_{12}^{(edge,~2)}},
\frac{D_{11}^{(bulk,~k)}}{D_{11}^{(bulk,~1)}}\mbox{ and } \frac{D_{12}^{(bulk,~k)}}{D_{12}^{(bulk,~2)}}
\eea
are
completely determined by the conjugacy classes of the corresponding kernels.
Therefore, the claim of the Corollary \ref{thm_edge_bulk} is an immediate consequence  of the following relations
between the kernels:
\bea
\lim_{R\rightarrow -\infty}  e^{(R+\lambdab)x}K_{11}^{(edge)}(R+x,R+\bar{x},R+y,R+\bar{y}\mid R+\lambda_1,R+\lambdab_1)e^{-(R+\lambdab)y}
\nonumber\\
=
K_{11}^{(bulk)}(x,\bar{x},y,\bar{y}\mid \lambda_1,\lambdab_1),
\nonumber\\
\label{ker10011}
\\
\lim_{R\rightarrow -\infty} K_{12}^{(edge)}(R+x,R+\bar{x},R+y,R+\bar{y}\mid R+\lambda_1,R+\lambdab_1,
R+\lambda_2,R+\lambdab_2)\nonumber\\
=K_{12}^{(bulk)}(x,\bar{x},y,\bar{y}\mid \lambda_1,\lambdab_1, \lambda_2,\lambdab_2),\nonumber\\
\label{ker10012}
\eea
Both (\ref{ker10011}) and (\ref{ker10012}) can be derived from the following asymptotic formula 
for the $H$-function (\ref{hedge}):
\bea\label{htokappa}
H(-\epsilon^{-1}+a,-\epsilon^{-1}+b,-\epsilon^{-1}+c,-\epsilon^{-1}+d,f)=
\left(e^{-f}-1+f\right)\left(1+O\left(\epsilon\right) \right),
\eea
where $\epsilon>0, (a,b,c,d,f) \in \C^5$. This formula follows directly from
the standard asymptotic for the complementary error function (\ref{eq_erfc}): 
for $a<0$,
\bea
F(a)=1+O\left(e^{-a^2/2}\right),~ F'(a)=O\left(e^{-a^2/2}\right).
\eea
The proof of (\ref{ker10011}):
\begin{eqnarray*}
&&\lim_{R\rightarrow -\infty} e^{(R+\lambdab)x}K_{11}^{(edge)}
(R+x,R+\bar{x},R+y,R+\bar{y}\mid R+\lambda,R+\lambdab)e^{-(R+\lambdab)y}\\
&=&-
\lim_{R\rightarrow -\infty}  e^{\lambdab x} \left( \frac{1+(x-\lambda)(\bar{x}-\lambdab)}{\pi z^2}
e^{\bar{x}(y-x)}
H\left(2R+a,2R+b,2R+c,2R+d,z\right)\mid_{z=(\bar{x}-\lambdab)(y-\lambda)}\right) e^{-\lambdab y}
\\
&=&  
\left.
\frac{1+(x-\lambda)(\bar{x}-\lambdab)}{\pi }e^{-|x-\lambda|^2}
\frac{1-(1-z)e^z}{z^2}\right|_{z=(\bar{x}-\lambdab)(y-\lambda)}
=K_{11}^{(bulk)}(x,\bar{x},y,\bar{y}\mid \lambda,\lambdab).
\end{eqnarray*}
To prove (\ref{ker10012}), let us first notice that (\ref{htokappa}) leads to
\begin{eqnarray}\label{kappaprep}
\kappa^{(edge)}(R+\bar{x},R+y\mid R+\lambda,R+\lambdab)=
e^{R^2+R(\bar{x}+y)-|\lambda|^2+\lambdab y+\lambda \bar{x}}
\kappa^{(bulk)}(\bar{x},y\mid \lambda,\lambdab)\left(1+O(R^{-1})\right),\nonumber\\
\end{eqnarray}
where $R<0.$ Also,
\bea\label{weightprep}
\omega^{(edge)}(R+x,R+\bar{x}\mid R+\lambda, R+\lambdab)
&=&\frac{1+(x-\lambda)(y-\lambdab)}{\pi}e^{-|R+x|^2}
\nonumber\\
&=&e^{-R^2-R(x+\bar{x})+\lambda\lambdab-\bar{x}\lambda-x\lambdab}
\omega^{(bulk)}(x,\bar{x}\mid \lambda, \lambdab).\quad
\eea
Notice that the last two relations are still valid if $\kappa^{(edge)}$ and $\omega^{(edge)}$ are treated as functions
on $\C^4$.
Substituting 
(\ref{kappaprep}), (\ref{weightprep}) into (\ref{eq_ker12edge}) we find
\bea
K_{12}^{(edge)}(R+x,R+\bar{x},R+y,R+\bar{y}\mid R+\lambda_1,R+\lambdab_1,R+\lambda_2,R+\lambdab_2)
\nonumber\\
=
e^{(R+\lambdab_2)(y-x)} K_{12}^{(bulk)}(x,\bar{x},y,\bar{y}\mid \lambda_1,\lambdab_1,\lambda_2,\lambdab_2)
\left(1+O(R^{-1})\right).
\eea
Therefore
\begin{eqnarray*}
\lim_{R\rightarrow -\infty} e^{(R+\lambdab_2)x}K_{12}^{(edge)}(R+x,R+\bar{x},R+y,R+\bar{y}\mid R+\lambda_1,R+\lambdab_1,
R+\lambda_2,R+\lambdab_2)e^{-(R+\lambdab_2)y} \nonumber\\
=K_{12}^{(bulk)}(x,\bar{x},y,\bar{y}\mid \lambda_1,\lambdab_1, \lambda_2,\lambdab_2).
\end{eqnarray*}
Equation (\ref{ker10012}) is established and therefore the
Corollary \ref{thm_edge_bulk} is proved.
\subsection{The proof of relations (\ref{prodd11}), (\ref{prodd12}) and Corollary \ref{thm_alg_dec}}\label{proof_decay}
We shall  start with deriving (\ref{prodd11}), (\ref{prodd12}).
According to Corollary \ref{thm_bulk},
\bea\label{dec1}
D_{11}^{(bulk,~k)}(\lv^{(k)})=\frac{1}{\pi^k} \prod_{i=2}^k \left(\frac{1+|\lambda_{i1}|^2}{|\lambda_{i1}|^4}e^{-|\lambda_{i1}|^2}\right) \det_{2\leq m,n\leq k} \left[1-(1-\bar{\lambda}_{m1}\lambda_{n1})
e^{\lambdab_{m1}\lambda_{n1}} \right].
\eea
Let $\mathbf{1}$ be a $(k-1)$-dimensional column vector with all components equal to $1$. Let $M, A$ be $(k-1) \times (k-1)$
matrix.  Let $\alpha$ be a constant. The next two identities follow 
directly from the block determinant formula \eqref{blockId}:
\bea
\det(M-\mathbf{1} \otimes \mathbf{1})&=&
\det \left(\begin{array}{c|c}
1&\mathbf{1}^T\\
\hline
\mathbf{1} &M
\end{array}
\right),
\\
\det \left(\begin{array}{c|c}
1&\mathbf{1}^T\\
\hline
\mathbf{1} &\alpha A
\end{array}
\right)
&=&\alpha^{k-2}
\det \left(\begin{array}{c|c}
\alpha&\mathbf{1}^T\\
\hline
\mathbf{1} &A
\end{array}
\right).
\eea
Applying the identities to the determinant in the r.h.s.  of (\ref{dec1}), one finds
\bea\label{dec2}
D_{11}^{(bulk,~k)}(\lv^{(k)})
&=&-\left(-\frac{1}{\pi}\right)^{k} \prod_{i=2}^k \left(\frac{1+|\lambda_{i1}|^2}{|\lambda_{i1}|^4}e^{-|\lambda_{i1}|^2}\right) 
\det_{2\leq m,n\leq k}\left(
\begin{array}{c|c}
1& \mathbf{1}^T\\
\hline
&\\
\mathbf{1} &(1-\bar{\lambda}_{m1}\lambda_{n1})
e^{\lambdab_{m1}\lambda_{n1}}
\end{array}\right)
\nonumber\\
&=&-\left(-\frac{1}{\pi}\right)^{k} \prod_{i=2}^k \left(\frac{1+|\lambda_{i1}|^2}{|\lambda_{i1}|^4}e^{-|\lambda_{i1}|^2+|\lambda_1|^2}\right) e^{-|\lambda_1|^2} \nonumber\\
&&\times \det_{2\leq m,n\leq k}\left(
\begin{array}{c|c}
e^{|\lambda_1|^2}& \mathbf{1}^T\\
\hline
&\\
\mathbf{1} &(1-\bar{\lambda}_{m1}\lambda_{n1})
e^{\lambdab_{m1}\lambda_{n1}-|\lambda_1|^2}
\end{array}\right).
\eea
Let
\bea
\frak{D}_m:=1-\lambda_{m1}\frac{\partial}{\partial \lambda_m}, m=1, \ldots k,
\eea
be a first order differential operator. Clearly, $$[\frak{D}_m,\frak{D}_n]=0,$$ for $m,n\geq 1$.
Observe also that 
\[
\frak{D}_pe^{\lambdab_{m1} \lambda_{n1}}=\left(1-\delta_{pn}\lambdab_{m1}\lambda_{n1}\right)
e^{\lambdab_{m1} \lambda_{n1}},~p,m,n \geq 1.
\]
These observations lead to the following identity:
\bea
\det
\left(
\begin{array}{c|c}
e^{|\lambda_1|^2}& \mathbf{1}^T\\
\hline
&\\
\mathbf{1} &(1-\bar{\lambda}_{m1}\lambda_{n1})
e^{\lambdab_{m1}\lambda_{n1}-|\lambda_1|^2}
\end{array}\right)
&=&\prod_{m=2}^k \frak{D}_{m}
\det
\left(
\begin{array}{c|c}
e^{|\lambda_1|^2}& \mathbf{1}^T\\
\hline
&\\
\mathbf{1} &
e^{\lambdab_{m1}\lambda_{n1}-|\lambda_1|^2}
\end{array}\right)
\nonumber\\
&=&\prod_{m=2}^k \left(\frak{D}_{m}e^{-\lambdab_m\lambda_1 -\lambdab_1\lambda_m}\right)
\det_{1\leq p,q \leq k}\left(
e^{\lambdab_p\lambda_q}
\right).\quad\quad\quad
\label{dec3}
\eea
Substituting (\ref{dec3}) into (\ref{dec2}) and simplifying we find:
\bea\label{dec4}
D_{11}^{(bulk,~k)}(\lv^{(k)})
=\left(-1\right)^{k-1}  \prod_{i=2}^k \left(\frac{1+|\lambda_{i1}|^2}{|\lambda_{i1}|^4}\right) 
\prod_{n=2}^k e^{-|\lambda_{n1}|^2} \frak{D_{n}}e^{|\lambda_{n1}|^2} \rho^{(bulk,~k)} (\lv^{(k)}),
\eea
where $\rho^{(bulk,~k)}$ is the bulk scaling limit of the $k$-point correlation function (\ref{cgin1}),
\bea
\rho^{(bulk,~k)}(\lv^{(k)})=\prod_{m=1}^k\frac{e^{-|\lambda_m|^2}}{\pi} \det_{1\leq i,j\leq k}
\left( e^{\lambdab_i \lambda_j}\right),
\eea
 see 
  Ginibre's classical paper \cite{ginibre}.
 Notice that
 \[
  e^{-|\lambda_{n1}|^2} \frak{D_{n}}e^{|\lambda_{n1}|^2}=\frak{D}_n-|\lambda_{n1}|^2.
 \]
 Substituting this formula into
 (\ref{dec4}) we arrive at the representation (\ref{prodd11}). 
 
 The easiest way to prove (\ref{prodd12})
 is to recall that the bulk scaling limits $D_{11}^{(bulk,~k)}$ and $D_{12}^{(bulk,~k)}$ are still related via 
 the formula (\ref{rltn}) of Lemma \ref{thm_rel}. Applying the relation to both sides of (\ref{prodd11}) and  noticing that
 \bea
 \hat{T}\rho^{(bulk,~k)}(\lv^{(k)})=-e^{|\lambda_{12}|^2} \rho^{(bulk,~k)}(\lv^{(k)}),
 \eea
 we obtain formula (\ref{prodd12}).
 
Now we are ready to prove Corollary \ref{thm_alg_dec}. The decay of
correlations for the complex Ginibre ensemble is Gaussian, 
\bea\label{as_rho}
\rho^{(bulk,~k)}(\lv^{(k)})=\pi^{-k}+O(e^{-|\lambda_{ij}|^2}, 1\leq i<j \leq k).
\eea
Therefore, for well separated eigenvalues, 
\bea\label{estm}
\frac{\partial}{\partial \lambda_m}\rho^{(bulk,~k)}(\lv^{(k)})=O(e^{-L^2}),~1\leq m\leq k,
\eea
where $L=\inf_{1\leq i\neq j \leq k} |\lambda_{ij}|$. Substituting (\ref{as_rho}) and (\ref{estm})
into (\ref{prodd11}), (\ref{prodd12}) we immediately arrive at the exact algebraic asymptotic
of conditional overlaps stated in Corollary \ref{thm_alg_dec}.

\section{Summary and Open Problems}\label{conc}

We have analysed the overlap between left and right eigenvectors in the complex Ginibre ensemble of random matrices, conditioned on $k$ complex eigenvalues. 
Starting from the results of Chalker and Mehlig we used a combination of the inversion of the moment matrix and theory of orthogonal polynomials in the complex plane,  to arrive at a determinantal structure for the diagonal overlap. It is valid for finite matrix size $N$ and fixed $k$ and is explicitly given in terms of a kernel, containing combinations of exponential polynomials. Its analyticity led us to deduce the off-diagonal overlap as a $k\times k$ determinant and its kernel as well. 

These findings allowed us to take the microscopic limit both in the bulk and the edge of the spectrum. Both bulk and edge kernel were explicitly derived and conjectured to be universal. 
For the bulk we restricted ourselves to the vicinity of the origin, but due to the translational invariance of the limiting kernel we expect to find the same answer everywhere in the bulk of the spectrum. 
At the edge we found a residual rotational symmetry of the kernel, that is independent of the angle where at the circular edge at $|z|=\sqrt{N}$ we take the limit.
At large argument separation the bulk limit answers 
also allowed us to derive the algebraic 
decay of the conditional overlaps and establish their asymptotic factorisation. 

It is an open question if the determinantal structures that we found can also be obtained in more general ensembles at finite-$N$, such as is products of Ginibre matrices, see \cite{BSV} for some conjectures (as well as \cite{Belinschi} for large-$N$), or in a more general non-Gaussian setting.  This will be a formidable task because, for example, for products of random matrices more and more off-diagonal elements of the moment matrix emerge. 

A further question is concerning other symmetry classes. Whilst the most difficult real Ginibre ensemble has been addressed in \cite{bourgade,fyodorov}, for the quaternionic Ginibre ensemble so far only first steps have been taken \cite{Dubach,AFK}. The structures found by Chalker and Mehlig persist, and the same result in the global, macroscopic regime is found  for the diagonal and off-diagonal overlap. Note however, that already in the symmetry class of complex matrices these are non-universal when going to non-Gaussian ensembles, cf. \cite{Belinschi}. It remains to be seen if a Pfaffian structure similar to the one in the present work can be found. In principle, the building block, a Pfaffian formula for expectation values of products of characteristic polynomials prevails \cite{AB}, that generalises the Christoffel type theorem \cite{akemann_vernizzi} from orthogonal to skew orthogonal polynomials in the complex plane.

Our motivation was, apart from finding integrable structures, the coupled 
stochastic motion of the complex eigenvalues and corresponding eigenvectors. Its further investigation is left for future work. 
\\

\noindent
{\bf Acknowledgements.} We would like to thank Yan Fyodorov and Mario Kieburg for fruitful discussions at an early stage of this project. Furthermore, partial support from the following grants is gratefully acknowledged:
CRC1283 ``Taming uncertainty and profiting from randomness and low regularity in analysis, stochastics and their applications"  by  
 the German Research Foundation DFG for Gernot Akemann, Roger Tribe is partially supported by a Leverhulme Research Fellowship RF-2-16-655.

\begin{appendix}
\section{Appendix: Normal Evolutions}\label{appendix}
The fact that  
Brownian motion with values in normal matrices is related to  
the $2$-dimensional log-gas seems to be a common albeit unpublished knowledge. Here we present
the verification of this fact just to emphasise the difference between the evolution equations
and the `Dysonian dynamics' for the complex Ginibre ensemble studied in \cite{grela}.

Let $M=\{H\in \C^{N^2}\mid [H,H^\dagger]=0\}$ be the space of $N\times N$
normal matrices. Notice that the space $M$ is not linear. It is not a smooth manifold
either - there are singularities at the points of degeneration of the spectrum of $H$.
It must still be possible to construct an $M$-valued Brownian motion, provided that
it does not visit the singular points. We can attempt to define this process by using
the Laplace-Beltrami operator (which exists on the complement to the set of singular points)
as the generator. Let $M_c$ be the complement to the set of singularities of $M$.
For that we need to choose good local coordinates on $M$. 
A reasonable choice would be
\bea
H=U\Lambda U^\dagger,
\eea
where $\Lambda$ is a complex diagonal matrix 
with entries $\lambda^1, \lambda^2, \ldots, \lambda^N$ and $U \in U(N)/U(1)^N$. 
The
Riemannian metric on $M_c$ is induced by the embedding $M_c\subset \C^{N^2}$,
\bea
G(\delta H, \delta H)=\tr \delta H \delta H^\dagger.
\eea
But for $H\in M$,
\bea
\delta H=U(\delta \Lambda+[\delta g,\Lambda])U^\dagger,
\eea
where $\delta g=U^\dagger \delta U \in Lie(U(N))$, with the additional restriction
$\delta g^{ii}=0$, $i=1,2,\ldots, N$. Therefore,
\bea
G(\delta H, \delta H)=
\sum_{i} |\delta \Lambda_{ii}|^2+2\sum_{i<j} |\lambda^i-\lambda^j|^2|\delta g^{ij}|^2.
\eea
(In this Section only we use superscripts to label $\lambda$'s in line with the labelling convention for local 
coordinates in differential geometry.)
We see that the metric tensor is diagonal in $(\Lambda, U)$-coordinates and that
\bea
\sqrt{\det(G)}=2^{\frac{N(N-1)}{2}}|\Delta^{(N)} (\Lambda)|^2.
\eea
The inversion of $G$ is straightforward and the generator of the Brownian motion on $M_c$
is given by
\bea
L=\frac{1}{|\Delta^{(N)} (\Lambda)|^2}\sum_{i=1}^N\left(\frac{\partial}{\partial \lambda^i}
|\Delta^{(N)} (\Lambda)|^2\frac{\partial}{\partial \overline{\lambda}^i}+
\frac{\partial}{\partial \overline{\lambda}^i}
|\Delta^{(N)} (\Lambda)|^2\frac{\partial}{\partial \lambda^i}+\frac{1}{2}\sum_{i<j}
\frac{1}{|\lambda^i-\lambda^j|^2}\frac{\partial}{\partial g^{ij}}\frac{\partial}{\partial \bar{g}^{ij}}
\right)\!\!.
\nonumber\\
\eea
There are two obvious points to notice: (i) The dynamics of eigenvalues is Markovian; (ii) Unitary degrees of freedom speed up near the singular points, i.e. at $\lambda^i=\lambda^j$ for some $i \neq j$ .

The generator for the eigenvalue dynamics is
\bea
L=\frac{1}{\Delta^{(N)} (\Lambda)}\sum_{i=1}^N\frac{\partial}{\partial \lambda^i}
\frac{\partial}{\partial \overline{\lambda}^i}
\Delta^{(N)} (\Lambda)+\frac{1}{\Delta^{(N)} (\overline{\Lambda})}\sum_{i=1}^N\frac{\partial}{\partial \lambda^i}
\frac{\partial}{\partial \overline{\lambda}^i}
\Delta^{(N)} (\bar{\Lambda}).
\eea

The corresponding system of SDE's is 
\bea\label{sden1}
d\lambda^i=2\sum_{k\neq i} \frac{1}{\overline{\lambda}^i-\overline{\lambda}^k} dt+\sqrt{2} dW^i_t,\\
\label{sden2}
d\lambdab^i=2\sum_{k\neq i} \frac{1}{\lambda^i-\lambda^k} dt+\sqrt{2} d\overline{W}^i_t,
\eea
where $\{W^i_t, \overline{W}^i_t\}_{t\geq 0}^{1\leq i\leq N}$ are independent complex Brownian motions
with non-zero covariances $\E\left(\bar{W}_t^i W_s^j\right)=\delta^{ij}s \wedge t$. Equations (\ref{sden1})
and (\ref{sden2}) are valid until the time of the first exit from $M_c$, which we conjecture to be infinite, but we do not
pursue the argument here. 

It is well known \cite{chau}, that the $t=1$ marginal distribution of eigenvalues for normal
evolutions (the so called normal random matrix model) coincides with the law (\ref{cginibrelaw}) for the complex
Ginibre eigenvalues. (This correspondence breaks down for models with 
non-Gaussian 
potentials.)  
Yet the evolution equations (\ref{sden1}, \ref{sden2}) are very obviously different from the equations for
the joint evolution of eigenvalues and eigenvectors for the complex Ginibre evolutions derived in \cite{grela}. 
\end{appendix}

\end{document}